\documentclass[twoside]{IEEEtran}

\usepackage{amsfonts,mathrsfs,amsmath,amsthm,amssymb}
\usepackage{mathtools}
\usepackage{xparse}
\usepackage{graphicx}
\usepackage{float}
\usepackage[dvipsnames]{xcolor}
\usepackage{colortbl}
\usepackage{multirow}
\usepackage{enumerate}   
\usepackage{booktabs}
\usepackage{cancel}
\usepackage{listings}
\usepackage{tikz}

\newcommand{\SU}[1]{\ensuremath{{\text{SU}(#1)}}}

\newcommand{\C}[1]{\ensuremath{\text{C}_{#1}}}

\newlength{\irrepwidth}
\newlength{\irrepbarthickness}
\newlength{\irrepbarheight}
\def\primes#1#2{\count0=#1 \loop \ifnum\count0>0 \advance\count0 by -1 #2\repeat}

\begingroup\lccode`~=`, \lowercase{\endgroup\def\comma{\mathcode`\,=\string"8000 \def~{\text{,}\,}}}
\begingroup\lccode`~=`, \lowercase{\endgroup}
\begingroup\lccode`~=`, \lowercase{\endgroup}
\newcommand\dynkincomma[1]{\ensuremath{(\comma #1)}}

\newcommand\weightalpha\dynkincomma
\newcommand\weightorthogonal\dynkincomma
\newcommand\rootorthogonal\dynkincomma

\allowdisplaybreaks

\theoremstyle{plain}
\newtheorem{theorem}{Theorem}
\newtheorem{corollary}{Corollary}
\newtheorem{lemma}{Lemma}
\newtheorem{proposition}{Proposition}

\theoremstyle{definition}
\newtheorem{definition}{Definition}
\newtheorem{remark}{Remark}

\newtheorem*{conjecture*}{Conjecture}

\newcommand{\bfit}[1]{\textbf{\textit{#1}}}

\newcommand{\CC}{\mathbb{C}} 
    
\DeclareMathAlphabet{\mathsfit}{T1}{\sfdefault}{\mddefault}{\sldefault}
\SetMathAlphabet{\mathsfit}{bold}{T1}{\sfdefault}{\bfdefault}{\sldefault}

\newcommand{\Z}{\mathsf{Z}}

\renewcommand{\SU}{\mathrm{SU}} 

    \newcommand{\supp}{\text{supp}}
    
    \renewcommand{\C}{\mathcal{C}}
    
\usepackage{listings}
\usepackage{xcolor}

\definecolor{codegreen}{rgb}{0,0.6,0}
\definecolor{codegray}{rgb}{0.2,0.2,0.2}
\definecolor{codepurple}{rgb}{0.58,0,0.82}
\definecolor{backcolour}{rgb}{0.95,0.9,0.99}

\lstdefinestyle{mystyle}{
    backgroundcolor=\color{white},   
    commentstyle=\color{codegreen},
    keywordstyle=[1]\color{magenta},
    keywordstyle=[2]\color{blue},
    otherkeywords = {PerfectGroup,IdGroup,Size,IdGroup,StructureDescription,MolienSeries,CharacterTable,Irr,IrreducibleRepresentationsDixon,Norm,ScalarProduct,gap>},
    morekeywords = [1]{PerfectGroup,IdGroup,Group,SmallGroup,Indicator,Display,DicyclicGroup,CyclicGroup,SymmetricParts,AntiSymmetricParts,ConjugacyClasses,Print,IrreducibleRepresentationsDixon,Image,CharacterTable,Irr,Norm,ScalarProduct,Degree},
    morekeywords = [2]{gap>,>},
    numberstyle=\tiny\color{codegray},
    stringstyle=\color{codepurple},
    basicstyle=\ttfamily\footnotesize,
    breakatwhitespace=false,         
    breaklines=true,                 
    captionpos=b,                    
    keepspaces=true,                 
    numbers=left,
    frame=single,
    numbersep=5pt,   
    frame = single, 
    showspaces=false,                
    showstringspaces=false,
    showtabs=false,  
    morecomment=[l][\color{codegreen}]{\#},
    alsoletter=?>,
    tabsize=2
}
\lstdefinestyle{mystyle2}{
    backgroundcolor=\color{white},   
    commentstyle=\color{codegreen},
    keywordstyle=[1]\color{magenta},
    keywordstyle=[2]\color{blue},
    otherkeywords = {PerfectGroup,IdGroup,Size,IdGroup,StructureDescription,MolienSeries,CharacterTable,Irr,IrreducibleRepresentationsDixon,Norm,ScalarProduct,gap>},
    morekeywords = [1]{PerfectGroup,IdGroup,Group,SmallGroup,Indicator,Display,DicyclicGroup,CyclicGroup,SymmetricParts,AntiSymmetricParts,ConjugacyClasses,Print,IrreducibleRepresentationsDixon,Image,CharacterTable,Irr,Norm,ScalarProduct,Degree},
    morekeywords = [2]{gap>,>},
    numberstyle=\tiny\color{codegray},
    stringstyle=\color{codepurple},
    basicstyle=\ttfamily\footnotesize,
    breakatwhitespace=false,         
    breaklines=true,                 
    captionpos=b,                    
    keepspaces=true,                 
    numbers=left,
    frame=single,
    numbersep=5pt,   
    frame = single, 
    showspaces=false,                
    showstringspaces=false,
    showtabs=false,  
    morecomment=[l][\color{codegreen}]{\#},
    alsoletter=?>,
    tabsize=2
}

\PassOptionsToPackage{hyphens}{url}\usepackage{hyperref}

\hypersetup{
    colorlinks=true,
    linkcolor=magenta!80!black,  
    citecolor=magenta!80!black,
    pdftitle={Minimum-Length Permutation-Invariant Qudit Codes from Edge-Colorings of Complete Graphs},
    pdfauthor={Eric Kubischta and Ian Teixeira},
    pdfsubject={Permutation-invariant quantum error correction},
    pdfkeywords={quantum codes, permutation invariance, qudits, edge coloring, bosonic loss},
}
\usepackage{url} 

\usepackage[capitalise]{cleveref}
\usepackage{physics}
\usepackage{bm}

\newcommand{\Sym}{\mathrm{Sym}}
\renewcommand{\L}{\mathscr{L}} 

\renewcommand{\Z}{\mathbb{Z}}	
\renewcommand{\CC}{\mathbb{C}}
\renewcommand{\C}{\mathcal{C}} 

\renewcommand{\L}{\mathscr{L}} 
 
\renewcommand{\H}{\mathcal{H}} 

\newcommand{\diag}{\text{diag}}

\newcommand{\Avec}{\bm{A}}
\newcommand{\Bvec}{\bm{B}}

\begin{document}

\title{Minimum-Length Permutation-Invariant Qudit Codes from Edge-Colorings of Complete Graphs}

\author{Eric~Kubischta, Ian~Teixeira

\thanks{Both authors contributed equally to this work.}
\thanks{Eric Kubischta is affiliated with the Department of Mathematics, Florida State University, Tallahassee, FL 32306.}
\thanks{Ian Teixeira is affiliated with the Department of Mathematics, University of California, San Diego, CA 92093}
\thanks{\copyright\ 2026 IEEE. Personal use of this material is permitted. Permission from IEEE must be obtained for all other uses, in any current or future media, including reprinting/republishing this material for advertising or promotional purposes, creating new collective works, for resale or redistribution to servers or lists, or reuse of any copyrighted component of this work in other works.}
}

\maketitle

\begin{abstract}
We study permutation-invariant quantum codes in the symmetric subspace $ \Sym^n(\mathbb C^q) $ of \(n\) qudits of local dimension \(q\). For every integer \(q\ge 2\), we
construct a permutation-invariant code with parameters $ ((4,q,2))_q. $
Thus four physical qudits suffice to encode one logical qudit with distance
two in the symmetric sector for every local dimension. We also show, using a
self-contained decoupling and rank argument, that no
permutation-invariant code of dimension \(q\) and distance at least \(2\) exists
in \(\Sym^n(\mathbb C^q)\) for \(n\le 3\). Hence four qudits are necessary and
sufficient. The construction has simple representation-theoretic and combinatorial
descriptions. On \(\Sym^4(\mathbb C^q)\), the distance-two
Knill--Laflamme conditions split into root and Cartan parts. By
restricting supports to the even-entry occupation layer, all root-error
conditions vanish automatically. The remaining Cartan conditions reduce to
linear balancing constraints on packets of occupation vectors. These packets
admit a natural graph-theoretic interpretation in terms of the vertices and
edges of the complete graph \(K_q\): for odd \(q\), they are organized by the
midpoint rule, while for even \(q\), they are organized by a decomposition of
\(K_q\) into perfect matchings. In this way, the existence of minimum-length
\(((4,q,2))_q\) permutation-invariant codes is reduced to a parity-dependent
edge-coloring problem on \(K_q\).
\end{abstract}


\tableofcontents

\section{Introduction}

Permutation-invariant quantum codes form a natural and highly structured class
of quantum error-correcting codes.  Their ambient Hilbert space is not the full
tensor product $(\CC^q)^{\otimes n} $
but rather the symmetric subspace $ \Sym^n(\CC^q), $ consisting of those \(n\)-qudit states invariant under permutations of the
tensor factors. This symmetry makes permutation-invariant codes naturally robust against
permutation errors and closely connects them to deletion errors, since deletion
of unknown positions is equivalent to erasure of fixed positions once the state
is symmetric \cite{PI2}.

Early work supplied symmetric-subspace constructions and the first systematic
group-theoretic treatment of PI codes.  It included explicit single-error-correcting
examples and families of non-additive codes
\cite{RuskaiPRL,2004permutation}.  Later work introduced scalable families of
permutation-invariant codes built from Dicke-state superpositions, extended these constructions to
multiple logical qubits and to qudits, and developed polynomial methods for
high-dimensional symmetric spaces \cite{PI2,PImore,PIqudit1}.

A parallel line of work connected permutation-invariant symmetry with physically
motivated noise models.  In particular, constant-excitation and bosonic variants
were developed for amplitude-damping noise \cite{OuyangChao}, while the relation
between PI codes and quantum deletion channels was made explicit in subsequent
work, which also gave efficient encoding and decoding procedures for certain PI
families \cite{PIdelete,ShibayamaHagiwara}.  This relation is particularly
natural in the present setting: \(\Sym^n(\CC^q)\) is also the sector with
exactly \(n\) particles in \(q\) bosonic modes.  In that realization, deletion of
one constituent becomes loss of one indistinguishable particle.  We make this
correspondence explicit in Section~\ref{sec:bosonic-interpretation}.
Operationally, this supplies a concrete loss-protection setting: four identical
photons distributed among \(q\) modes, or four \(q\)-level neutral atoms under
collective control.  In either case a loss can be corrected without identifying
which constituent was removed.  The block length remains four as the internal
or mode dimension \(q\) grows.

More recent work has highlighted additional structure and applications of permutation-invariant
qubit codes.  Certain families admit nontrivial transversal logical gates
\cite{us2}.  New explicit constructions improved known parameter tradeoffs
\cite{PI3}.  Other work has constructed permutation-invariant qubit
codes using Tverberg partitions
\cite{BumgardnerWstarMetricThesis,ShorsLieDetectionREU,XuClassicalConstructionREU,ConvexPI}.
Recent work has also established a general theory of error correction for
permutation-invariant codes, including efficient symmetric-group-based
algorithms that correct any correctable error
\cite{ouyangPIerrorcorrection}.

Uy and Gangloff construct qudit codes inside symmetric-power representations
\(\Sym^N(\CC^d)\), using Heisenberg--Weyl covariance and Knill--Laflamme
reductions for collective \(\mathfrak{su}(d)\) errors \cite{UyGangloff}.
Bond et al. extend the deletion-error conditions to qudit PI codes and study
block-length scaling numerically, allowing the logical and physical local
dimensions to differ \cite{BondEtAlPINumerical}.  Their qubit study varies the
code distance, while their qudit numerics fix \(d=3\) and vary the logical and
physical dimensions.  Our focus is the sharp \(d=2\) problem with equal
logical and physical dimensions: we determine the exact minimum block length
analytically and give a closed-form construction for every local dimension.

Two small members of our family were known previously.  When \(q=2\), the
construction is exactly the four-qubit single-deletion code of Hagiwara and
Nakayama, up to exchanging the two logical basis vectors; their work also
established its minimum length in the binary setting
\cite{Hagiwara2020FourQubitDeletion}.  When \(q=3\), the construction is the
error-detecting code in the \((4,0)\) representation of \(\SU(3)\) found by
Herbert, Gross, and Newman
\cite{herbert2023qutritcodesrepresentationssu3}.  Our contribution is a uniform
construction for all \(q\), including all \(q\ge4\), the complete-graph
mechanism that explains both parity classes, and a proof of minimum block
length throughout the PI sector.

Our first main result is the existence of a distance-two permutation-invariant
encoding of one logical qudit into four physical qudits for every local
dimension \(q\).  More precisely, Theorems~\ref{thm:general-family} and
\ref{thm:even-family} construct a code
\[
    \mathcal C\subseteq \Sym^4(\CC^q)
\]
with parameters \(((4,q,2))_q\) for every \(q\ge 2\), treating odd and even
\(q\) separately.

For odd \(q\), the construction is especially explicit.  The logical basis is
generated by cyclic translates
\[
    \ket{\mathbf r}=(X^{\otimes 4})^r\ket{\mathbf 0},
    \qquad r\in\mathbb Z_q,
\]
where
\[
    \ket{\mathbf 0}
    =
    \frac{1}{\sqrt q}\ket{0000}
    +
    \frac{1}{\sqrt{3q}}
    \sum_{i=1}^{(q-1)/2}
    \ket{\overline{i\,i\,q-i\,q-i}}.
\]
Moreover, this odd-\(q\) family has a transversal logical Pauli action:
\(X^{\otimes4}\) acts as \(X_L\) and \(Z^{\otimes4}\) acts as
\(Z_L^{\,4}\); see Proposition~\ref{prop:odd-transversal-pauli}.

Our second main result is minimum length.  Theorem~\ref{thm:no-small-length}
shows that no permutation-invariant code of dimension \(q\) and distance at
least \(2\) exists in \(\Sym^n(\CC^q)\) for \(n\le 3\).  Together with the
construction, this proves that four physical qudits are necessary and
sufficient for a distance-two permutation-invariant encoding of one logical
qudit of local dimension \(q\).

Beyond the existence result, the construction is useful because it isolates a
simple representation-theoretic mechanism: in \(\Sym^4(\CC^q)\), the
distance-two Knill--Laflamme conditions split into root conditions, controlled
by support separation, and Cartan conditions, controlled by linear balancing. The symmetric subspace
\[
    \Sym^4(\CC^q)
\]
is the irreducible \(\SU(q)\)-module of highest weight \((4,0,\dots,0)\), and
its natural basis is indexed by occupation vectors, equivalently by weights.
Relative to the standard root-space decomposition of \(\mathfrak{sl}_q\), the
distance-two Knill--Laflamme conditions split into two types.

First, the root operators \(E_{ab}\) move weight by a single root.  By
restricting supports to the even-entry occupation layer, all root-error
matrix elements vanish automatically.  For \(n=4\), this layer consists
exactly of the occupation vectors
\[
    4e_i
    \qquad\text{and}\qquad
    2e_i+2e_j,
\]
which admit a natural graph-theoretic interpretation as the vertices and edges
of the complete graph \(K_q\).

Second, the Cartan operators act diagonally in the weight basis.  Once the
supports of the logical basis states are chosen to be disjoint, the remaining
Knill--Laflamme conditions reduce to diagonal expectation constraints.  In
graph-theoretic terms, one must organize the vertices and edges of \(K_q\) into
\(q\) disjoint packets satisfying linear balancing conditions determined by the
Cartan subalgebra.

The parity of \(q\) determines how this combinatorial organization is achieved.
For odd \(q\), the edges of \(K_q\) are grouped by the midpoint rule
\[
    \{i,j\}\mapsto \frac{i+j}{2}\pmod q,
\]
which is well-defined because \(2\) is invertible modulo \(q\).  This produces
a matching on the \(q-1\) vertices other than the midpoint, together with the
omitted vertex itself.  For even \(q\), the midpoint rule fails, and one
instead uses a decomposition of \(E(K_q)\) into \(q-1\) perfect matchings,
equivalently, a \(1\)-factorization of \(K_q\), together with a separate vertex
packet.  In both cases, the packets lie entirely in the even-entry layer, so
root separation is automatic and the construction reduces to Cartan balance.

Thus the existence of minimum-length \(((4,q,2))_q\) permutation-invariant codes is
governed by a parity-dependent edge-coloring problem on the complete graph
\(K_q\).  From the coding-theoretic point of view, this yields a uniform
family of optimal block-length constructions.
From the representation-theoretic point of view, it provides a concrete design
principle: root errors are controlled combinatorially by support separation,
while Cartan errors are controlled by linear balancing conditions on the
squared amplitudes.

\section{Preliminaries}
\label{sec:preliminaries}

Let $ \H := (\CC^q)^{\otimes n} $ be the Hilbert space of \(n\) qudits of local dimension \(q\).  The symmetric
group \(S_n\) acts on \(\H\) by permuting tensor factors, and the
permutation-invariant subspace is $ V:=\Sym^n(\CC^q)\subseteq \H. $ A \bfit{permutation-invariant quantum code} is a subspace $  \C\subseteq V. $ We say that \(\C\) has parameters \(((n,K,d))_q\) if it has
dimension \(K\) and distance \(d\).  In this paper we are primarily interested
in the case $ K=q $ and $ d=2 $. That is, permutation-invariant encodings of one logical qudit with distance
two. The dimension of the symmetric subspace is
\[
    \dim V=\dim \Sym^n(\CC^q)=\binom{n+q-1}{n}.
\]

\subsection{Distance-two Knill--Laflamme conditions}

Let \(P\) denote the orthogonal projector onto the code \(\C\). A code has distance at least \(2\) if and only if it detects all single-site
errors.  In the permutation-invariant setting, these conditions may be checked
using collective traceless one-body operators.

For \(A\in\L(\CC^q)\), define its collective action on \(n\) qudits by
\[
    A^{(n)}
    :=
    \sum_{r=1}^n
    I^{\otimes(r-1)}\otimes A\otimes I^{\otimes(n-r)}.
\]
We will apply this construction to
\[
    \mathfrak{sl}_q\subset \L(\CC^q),
\]
viewed as the traceless linear operators on \(\CC^q\).

\begin{lemma}
\label{lem:collective-errors-suffice}
Let \(\C\subseteq \Sym^n(\CC^q)\) be a permutation-invariant code.  Then
\(\C\) detects all single-site errors if and only if
\[
    P A^{(n)}P=\lambda_A P
    \qquad
    \text{for all }A\in\mathfrak{sl}_q
\]
for some scalars \(\lambda_A\in\CC\).
\end{lemma}

\begin{proof}
Let
\[
    A^{[r]}
    :=
    I^{\otimes(r-1)}\otimes A\otimes I^{\otimes(n-r)}
\]
denote the operator \(A\) acting on the \(r\)-th tensor factor.  If
\(\ket{\psi},\ket{\phi}\in\Sym^n(\CC^q)\), then
\[
    \bra{\psi}A^{[r]}\ket{\phi}
\]
is independent of \(r\), because both vectors are invariant under
permutations of the tensor factors.  Hence
\[
    \bra{\psi}A^{[r]}\ket{\phi}
    =
    \frac1n
    \bra{\psi}A^{(n)}\ket{\phi}.
\]
Thus the Knill--Laflamme matrix elements for single-site traceless errors are
equivalent to the Knill--Laflamme matrix elements for the collective operators
\(A^{(n)}\).

The identity component of a single-site operator contributes only a scalar
multiple of the identity on the code.  Therefore it suffices to check
traceless single-site operators, i.e.\ \(A\in\mathfrak{sl}_q\).
\end{proof}

By Lemma~\ref{lem:collective-errors-suffice}, the code
\(\C\subseteq V\) has distance at least \(2\) if and only if
\begin{equation}
\label{eq:KL-distance-two-linear}
    P A^{(n)}P=\lambda_A P
    \qquad
    \text{for all }A\in\mathfrak{sl}_q,
\end{equation}
for some scalars \(\lambda_A\in\CC\).  Equivalently, for any orthonormal basis
\(\{\ket{\psi_i}\}\) of \(\C\),
\begin{equation}
\label{eq:KL-matrix-elements}
    \bra{\psi_i}A^{(n)}\ket{\psi_j}
    =
    \lambda_A\delta_{ij}
    \qquad
    \text{for all }A\in\mathfrak{sl}_q.
\end{equation}
We use \eqref{eq:KL-distance-two-linear} as the projector form of the
distance-two conditions and \eqref{eq:KL-matrix-elements} as their
matrix-element form throughout.

We verify these conditions using the standard Cartan--Weyl basis of
\(\mathfrak{sl}_q\).  For \(a\neq b\), let \(E_{ab}\) denote the matrix unit
in \(\L(\CC^q)\) sending \(\ket b\mapsto \ket a\).  These are the
\bfit{root operators}.  For the Cartan subalgebra, we use the diagonal
traceless basis
\[
    H_i:=E_{i-1,i-1}-E_{i,i},
    \qquad
    1\le i\le q-1.
\]
We refer to the \(H_i\) as the \bfit{Cartan operators}.  When these one-body
operators act on \(V=\Sym^n(\CC^q)\), we use the same symbols for their
collective actions.

Accordingly, the distance-two conditions split into two types: those involving
root operators, which move between weights, and those involving Cartan
operators, which act diagonally.

\subsection{Occupation-number basis and weights}

An orthonormal basis of \(V=\Sym^n(\CC^q)\) is given by the occupation-number
states
\[
    \ket{a_0,\dots,a_{q-1}},
    \qquad
    a_0+\cdots+a_{q-1}=n,
\]
defined as the normalized symmetrization of a computational basis word
containing \(a_j\) copies of the symbol \(j\).

For example, when \(q=3\) and \(n=4\),
\[
    \ket{2,0,2}
    =
    \frac{1}{\sqrt6}\ket{\overline{0022}},
\]
where \(\ket{\overline{0022}}\) denotes the unnormalized sum of the six
distinct permutations of the word \(0022\).  More generally, throughout the
paper, an overline denotes the unnormalized sum over all distinct
permutations of the indicated word.

These occupation-number states are weight vectors for the Cartan subalgebra.

\begin{lemma}
\label{lem:occupation-weight-dictionary}
Let
\[
    \ket{a_0,\dots,a_{q-1}}\in \Sym^n(\CC^q).
\]
Then its weight with respect to the Cartan basis
\[
    H_i=E_{i-1,i-1}-E_{i,i},
    \qquad
    1\le i\le q-1,
\]
is
\[
    (a_0-a_1,\;a_1-a_2,\;\dots,\;a_{q-2}-a_{q-1}).
\]
Equivalently,
\[
    \begin{aligned}
    H_i\ket{a_0,\dots,a_{q-1}}
    &= (a_{i-1}-a_i)\ket{a_0,\dots,a_{q-1}}\\
    &\quad (1\le i\le q-1),
    \end{aligned}
\]
where \(H_i\) denotes the collective action on \(\Sym^n(\CC^q)\).
\end{lemma}

\begin{proof}
The one-body operator \(H_i\) acts on \(\CC^q\) as \(+1\) on \(\ket{i-1}\),
as \(-1\) on \(\ket i\), and as \(0\) on all other computational basis
vectors.  Therefore its collective action on an occupation-number state counts
the number of occurrences of \(i-1\) minus the number of occurrences of \(i\).
This gives the eigenvalue \(a_{i-1}-a_i\).
\end{proof}

In particular, for \(q=3\), the occupation-number state
\(\ket{a_0,a_1,a_2}\) has weight
\[
    (a_0-a_1,\;a_1-a_2).
\]

Conversely, when \(q=3\), a weight \((\lambda_1,\lambda_2)\) in
\(\Sym^n(\CC^3)\) corresponds to the occupation numbers
\[
    \begin{aligned}
    a_0&=\frac{n+2\lambda_1+\lambda_2}{3},\\
    a_1&=\frac{n-\lambda_1+\lambda_2}{3},\\
    a_2&=\frac{n-\lambda_1-2\lambda_2}{3}.
    \end{aligned}
\]
For \(a\neq b\), the collective root operator \(E_{ab}\) acts by moving one
unit of occupation from \(b\) to \(a\):
\[
    E_{ab}\ket{a_0,\dots,a_{q-1}}
    \propto
    \ket{a_0,\dots,a_a+1,\dots,a_b-1,\dots,a_{q-1}},
\]
whenever \(a_b>0\), and annihilates the state otherwise.  Thus root operators
change occupation by a single elementary move, and hence connect weights
differing by a single root.

Given a permutation-invariant state of $ n $ qudits of local dimension $ q $ 
\[
    \ket{\psi}
    =
    \sum_{a_0+\cdots+a_{q-1}=n}
    c_{a_0,\dots,a_{q-1}}
    \ket{a_0,\dots,a_{q-1}},
\]
we define the \bfit{support} of \(\ket{\psi}\) by
\[
    \supp\ket{\psi}
    :=
    \left\{
    (a_0,\dots,a_{q-1})
    :
    c_{a_0,\dots,a_{q-1}}\neq 0
    \right\}.
\]
Equivalently, via Lemma~\ref{lem:occupation-weight-dictionary}, we may regard
the support as a subset of the weight lattice of \(V\).  This is the point of
view used in the next section.

\subsection{Cyclic shift symmetry}

Let \(X\) denote the cyclic shift on the computational basis of \(\CC^q\):
\[
    X\ket a=\ket{a+1\!\!\!\pmod q}.
\]
The operator \(X^{\otimes n}\) acts by adding \(1\bmod q\) to each tensor
factor.  Consequently, on occupation-number states it cyclically permutes the
entries:
\[
    X^{\otimes n}\ket{a_0,a_1,\dots,a_{q-1}}
    =
    \ket{a_{q-1},a_0,\dots,a_{q-2}}.
\]

This symmetry will be used in the odd-\(q\) construction to generate all
logical basis vectors from a single seed vector.  It also interacts well with
the Cartan subalgebra.

\begin{lemma}
\label{lem:X-normalizes-Cartan}
Let \(\mathfrak h=\mathrm{span}\{H_1,\dots,H_{q-1}\}\) be the Cartan
subalgebra of diagonal traceless matrices.  Then
\[
    (X^{\otimes n})^\dagger \mathfrak h X^{\otimes n}
    =
    \mathfrak h.
\]
\end{lemma}

\begin{proof}
The operator \(X\) is a permutation matrix in the computational basis.
Conjugating a diagonal matrix by a permutation matrix again gives a diagonal
matrix, and trace is preserved under conjugation.  Hence, for every
\(H\in\mathfrak h\),
\[
    (X^{\otimes n})^\dagger H X^{\otimes n}
\]
is again diagonal and traceless.  Therefore it lies in \(\mathfrak h\).
\end{proof}

Thus, once the diagonal Knill--Laflamme conditions are verified for one
logical basis vector, they may be transported to its cyclic translates.

\subsection{Fixed-particle bosonic interpretation}
\label{sec:bosonic-interpretation}

The occupation-number description above has an exact second-quantized
interpretation.  Let \(\mathcal F_n^{(q)}\) be the \(n\)-particle sector of
\(q\) bosonic modes, with normalized Fock basis
\[
    \ket{a_0,\dots,a_{q-1}}_{\mathrm F},
    \qquad \sum_i a_i=n.
\]
The identification of occupation-number bases gives a unitary isomorphism
\[
    \Sym^n(\CC^q)\cong \mathcal F_n^{(q)}.
\]
Under this isomorphism, the collective matrix unit is the number-preserving
bilinear
\[
    E_{ab}^{(n)}\longleftrightarrow
    \hat a_a^\dagger \hat a_b,
\]
where \(\hat a_b\) annihilates a particle in mode \(b\).  The error operators
for loss of one particle are \(\{\hat a_b\}_{b=0}^{q-1}\), and their
Knill--Laflamme products are precisely
\[
    \hat a_a^\dagger\hat a_b.
\]
Consequently, the collective one-body conditions
\[
    P E_{ab}^{(n)}P=c_{ab}P
\]
are exactly the conditions for correcting one particle-loss event in the
fixed-particle bosonic realization.  In first-quantized language, this is a
single deletion: the lost constituent carries no position label identifying
which tensor factor was removed.  Thus the length-four codes constructed
below may equivalently be viewed as four-particle, \(q\)-mode bosonic loss
codes.

This formulation is directly relevant to multimode photonic encodings: total
particle number can in principle herald that one photon was lost, while the
Knill--Laflamme conditions ensure that recovery need not identify the lost
photon or its mode.  The first-quantized formulation similarly applies to a
permutation-symmetric ensemble of \(q\)-level neutral atoms when control and
readout are collective and the missing atom is not site-resolved.  Designing
platform-specific preparation and recovery circuits is a separate engineering
problem; the present result fixes the smallest symmetric block on which such a
scheme can be based.

\section{A root--weight separation principle}
\label{sec:root-separation}

In this section we formulate a general principle for constructing
distance-two permutation-invariant codes.  The key observation is that the
action of \(\mathfrak{sl}_q\) decomposes into root operators, which move
between weights, and Cartan operators, which act diagonally.  By selecting
supports in the weight lattice that are sufficiently separated with respect
to roots, all root-error Knill--Laflamme conditions are automatically
satisfied.  The remaining Cartan conditions reduce to elementary linear
balancing constraints on the squared moduli of the coefficients. Related support-selection methods in the context of quantum metric spaces were
studied in~\cite{BumgardnerWstarMetricThesis,ShorsLieDetectionREU,XuClassicalConstructionREU}.

\subsection{Weights and root adjacency}

Let
\[
V=\Sym^n(\CC^q)
\]
be viewed as a representation of \(\mathfrak{sl}_q\).  The occupation-number
basis
\[
\ket{a_0,\dots,a_{q-1}},
\qquad
a_0+\cdots+a_{q-1}=n,
\]
is a weight basis for the standard Cartan subalgebra \(\mathfrak h\) of
diagonal traceless matrices.

For \(a\neq b\), let \(E_{ab}\) denote the root operator which sends
\(\ket b\) to \(\ket a\).  On occupation-number states, the collective action
of \(E_{ab}\) moves one unit of occupation from \(b\) to \(a\):
\[
\begin{gathered}
E_{ab}^{(n)}\ket{a_0,\dots,a_{q-1}}\propto\\
\ket{a_0,\dots,a_a+1,\dots,a_b-1,\dots,a_{q-1}},
\end{gathered}
\]
whenever \(a_b>0\), and annihilates the state otherwise.  Thus a root
operator connects two weights precisely when their difference is a root.

\begin{definition}
Two weights of \(V\) are said to be \emph{root-adjacent} if their difference
is a root of \(\mathfrak{sl}_q\).
\end{definition}

Equivalently, two occupation vectors are root-adjacent if one can be obtained
from the other by moving one unit of occupation from one coordinate to
another.  Thus the weight set of \(V\) carries a natural graph structure, with
edges corresponding to root operators.

\subsection{Root-separated supports}

\begin{definition}
Let \(\mathcal S\) and \(\mathcal T\) be subsets of the weight set of \(V\).

\begin{enumerate}
\item The set \(\mathcal S\) is \emph{root-separated} if no two distinct
weights in \(\mathcal S\) are root-adjacent.

\item The sets \(\mathcal S\) and \(\mathcal T\) are \emph{mutually
root-separated} if no weight in \(\mathcal S\) is root-adjacent to any weight
in \(\mathcal T\).
\end{enumerate}
\end{definition}

Thus a root-separated set contains no edge of the weight graph, and mutually
root-separated sets have no root edge between them.

\subsection{Vanishing of root-error matrix elements}

The basic consequence of root separation is the following.

\begin{proposition}
\label{prop:root-vanishing-weight}
Let
\[
\ket{\psi}=\sum_{\lambda\in\mathcal S}c_\lambda\ket{\lambda},
\qquad
\ket{\phi}=\sum_{\mu\in\mathcal T}d_\mu\ket{\mu}
\]
be states in \(V\).

\begin{enumerate}
\item If \(\mathcal S\) is root-separated, then root operators satisfy the diagonal Knill-Laflamme conditions, that is,
\[
\bra{\psi}E_\alpha\ket{\psi}=0
\qquad
\text{for every root operator } E_\alpha.
\]

\item If \(\mathcal S\) and \(\mathcal T\) are mutually root-separated, then root operators satisfy the off-diagonal Knill-Laflamme conditions, that is,
\[
\bra{\psi}E_\alpha\ket{\phi}=0
\qquad
\text{for every root operator } E_\alpha.
\]
\end{enumerate}
\end{proposition}

\begin{proof}
Expanding in the weight basis gives
\[
\bra{\psi}E_\alpha\ket{\phi}
=
\sum_{\lambda\in\mathcal S}
\sum_{\mu\in\mathcal T}
\overline{c_\lambda}d_\mu
\bra{\lambda}E_\alpha\ket{\mu}.
\]
The matrix element \(\bra{\lambda}E_\alpha\ket{\mu}\) can be nonzero only if
\[
\lambda=\mu+\alpha,
\]
that is, only if \(\lambda\) and \(\mu\) are root-adjacent.  The stated
root-separation assumptions exclude all such pairs, so the sum vanishes.
The first assertion is the special case \(\mathcal T=\mathcal S\) and
\(\ket{\phi}=\ket{\psi}\).
\end{proof}

\subsection{A support-separation criterion}

We now combine the preceding root-vanishing observation with the diagonal
action of the Cartan subalgebra.

Let
\[
\ket{\psi_0},\dots,\ket{\psi_{K-1}}\in V
\]
be an orthonormal family, and let
\[
\mathcal S_i:=\supp\ket{\psi_i}.
\]
If the supports are pairwise disjoint, then Cartan operators have no
off-diagonal matrix elements between distinct codewords, because Cartan
operators are diagonal in the weight basis.  Hence, after root separation has
removed the root-error matrix elements, the only remaining condition is that
the diagonal Cartan expectations be independent of the logical basis state.

\begin{theorem}[Support-separation criterion]
\label{thm:support-separation}
Let
\[
\ket{\psi_0},\dots,\ket{\psi_{K-1}}\in \Sym^n(\CC^q)
\]
be an orthonormal family of states, and let
\[
\mathcal S_i:=\supp\ket{\psi_i}.
\]
Suppose that the following conditions hold.

\begin{enumerate}
\item The supports \(\mathcal S_0,\dots,\mathcal S_{K-1}\) are pairwise
disjoint.

\item Each support \(\mathcal S_i\) is root-separated.

\item The supports \(\mathcal S_i\) and \(\mathcal S_j\) are mutually
root-separated whenever \(i\neq j\).

\item For every \(H\in\mathfrak h\), the scalar
\[
\bra{\psi_i}H\ket{\psi_i}
\]
is independent of \(i\).
\end{enumerate}

Then
\[
\mathcal C:=\mathrm{span}\{\ket{\psi_0},\dots,\ket{\psi_{K-1}}\}
\]
detects all traceless one-body errors.  In particular, \(\mathcal C\) has
distance at least \(2\).
\end{theorem}

\begin{proof}
Let \(P\) be the orthogonal projector onto \(\mathcal C\).  It is enough to
verify the Knill--Laflamme conditions for the Cartan--Weyl basis of
\(\mathfrak{sl}_q\).

First let \(E_\alpha\) be a root operator.  By
Proposition~\ref{prop:root-vanishing-weight} and the root-separation
assumptions,
\[
\bra{\psi_i}E_\alpha\ket{\psi_j}=0
\qquad
\text{for all } i,j.
\]
Therefore
\[
P E_\alpha P=0.
\]

Now let \(H\in\mathfrak h\).  Since \(H\) acts diagonally in the weight basis
and the supports are pairwise disjoint, all off-diagonal Cartan matrix
elements vanish:
\[
\bra{\psi_i}H\ket{\psi_j}=0
\qquad
(i\neq j).
\]
By assumption, the diagonal matrix elements are independent of \(i\).  Hence
there exists a scalar \(\lambda_H\) such that
\[
\bra{\psi_i}H\ket{\psi_i}=\lambda_H
\qquad
\text{for all } i.
\]
Thus
\[
P H P=\lambda_H P.
\]

The Cartan and root operators span \(\mathfrak{sl}_q\).  Hence
\[
P A^{(n)}P=\lambda_A P
\qquad
\text{for every } A\in\mathfrak{sl}_q,
\]
and so \(\mathcal C\) detects all traceless one-body errors.  Therefore
\(\mathcal C\) has distance at least \(2\).
\end{proof}

\subsection{The cyclic specialization}

The preceding criterion is especially useful when the codewords are generated
from a single vector by the cyclic shift.

Let \(X\) denote the cyclic shift
\[
X\ket a=\ket{a+1\!\!\!\pmod q}.
\]
As noted above, \(X\) normalizes the Cartan subalgebra: conjugating a diagonal
traceless matrix by \(X\) again gives a diagonal traceless matrix.  Thus
\[
(X^{\otimes n})^\dagger \mathfrak h X^{\otimes n}=\mathfrak h.
\]

\begin{corollary}[Cyclic support-separation criterion]
\label{cor:cyclic-separation}
Let \(\ket{\psi_0}\in V=\Sym^n(\CC^q)\), and define
\[
\ket{\psi_r}:=(X^{\otimes n})^r\ket{\psi_0},
\qquad
r=0,1,\dots,q-1.
\]
Let
\[
\mathcal S_r:=\supp\ket{\psi_r}.
\]
Suppose that the supports \(\mathcal S_0,\dots,\mathcal S_{q-1}\) are
pairwise disjoint, each \(\mathcal S_r\) is root-separated, the supports
\(\mathcal S_r\) and \(\mathcal S_s\) are mutually root-separated whenever
\(r\neq s\), and
\[
\bra{\psi_0}H\ket{\psi_0}=0
\qquad
\text{for all } H\in\mathfrak h.
\]
Then
\[
\mathcal C:=\mathrm{span}\{\ket{\psi_0},\dots,\ket{\psi_{q-1}}\}
\]
is a permutation-invariant quantum code with distance at least \(2\).
\end{corollary}

\begin{proof}
We verify the Cartan condition in Theorem~\ref{thm:support-separation}.  For
any \(H\in\mathfrak h\),
\[
\bra{\psi_r}H\ket{\psi_r}
=
\bra{\psi_0}
(X^{\otimes n})^{-r}H(X^{\otimes n})^r
\ket{\psi_0}.
\]
Since \(X\) normalizes the Cartan subalgebra, the conjugated operator
\[
(X^{\otimes n})^{-r}H(X^{\otimes n})^r
\]
lies again in \(\mathfrak h\).  By assumption, \(\ket{\psi_0}\) has zero
expectation against every Cartan element.  Therefore
\[
\bra{\psi_r}H\ket{\psi_r}=0
\qquad
\text{for all } r.
\]
The hypotheses of Theorem~\ref{thm:support-separation} are now satisfied,
with \(\lambda_H=0\) for every \(H\in\mathfrak h\).  The result follows.
\end{proof}

\subsection{Reduction to a linear balancing problem}

In applications, the Cartan condition in
Theorem~\ref{thm:support-separation} becomes a finite system of linear
equations.

Suppose
\[
\ket{\psi_i}
=
\sum_{\lambda\in\mathcal S_i}
c_{i,\lambda}\ket{\lambda}.
\]
Since \(H\in\mathfrak h\) acts diagonally in the weight basis,
\[
H\ket{\lambda}=\lambda(H)\ket{\lambda},
\]
we have
\[
\bra{\psi_i}H\ket{\psi_i}
=
\sum_{\lambda\in\mathcal S_i}
\lambda(H)\,|c_{i,\lambda}|^2.
\]
Thus the Cartan constraints depend only on the nonnegative weights
\[
a_{i,\lambda}:=|c_{i,\lambda}|^2.
\]

In particular, the condition that the Cartan expectation be independent of
\(i\) becomes
\[
\sum_{\lambda\in\mathcal S_i}
\lambda(H)a_{i,\lambda}
=
\sum_{\mu\in\mathcal S_j}
\mu(H)a_{j,\mu}
\qquad
\text{for all } H\in\mathfrak h
\]
and all \(i,j\), together with the normalization conditions
\[
a_{i,\lambda}\ge 0,
\qquad
\sum_{\lambda\in\mathcal S_i}a_{i,\lambda}=1.
\]

For the cyclic odd-\(q\) construction, it is enough to solve these equations
for a single packet \(\mathcal S_0\), because the cyclic shift transports the
Cartan conditions to all other packets.  For the even-\(q\) construction, the
same balancing condition is checked directly for each packet.

Thus the construction of distance-two permutation-invariant codes reduces to
two tasks: choose mutually root-separated supports in the weight lattice, and
solve the resulting linear balancing equations for the Cartan subalgebra.

\section{The combinatorial support model}
\label{sec:combinatorial-model}

In this section we specialize the support-separation criterion to
\[
    V=\Sym^4(\CC^q).
\]
The goal is to isolate a simple class of supports for which all root-error
Knill--Laflamme conditions vanish automatically.  The remaining Cartan
conditions then become elementary balancing constraints on the vertices and
edges of a complete graph.
\subsection{The even-entry layer}

Let \(e_0,\dots,e_{q-1}\) denote the standard basis vectors of \(\mathbb Z^q\).
An occupation vector
\[
    a=(a_0,\dots,a_{q-1}),
    \qquad
    a_0+\cdots+a_{q-1}=4,
\]
is said to lie in the \emph{even-entry layer} if every \(a_i\) is even.
Since the total occupation is \(4\), there are exactly two types of such
vectors:
\[
    4e_i,
    \qquad
    2e_i+2e_j
    \quad (i\neq j).
\]
We denote the set of all even-entry occupation vectors by
\[
\mathcal E
:=
\{\,4e_i : 0\le i\le q-1\,\}
\;\cup\;
\{\,2e_i+2e_j : 0\le i<j\le q-1\,\}.
\]

The key feature of \(\mathcal E\) is that it contains no root-adjacent pairs.

\begin{lemma}
\label{lem:even-layer-root-sep}
If \(x,y\in\mathcal E\) and \(x\neq y\), then
\[
    x-y\neq e_a-e_b
    \qquad
    \text{for all } a\neq b.
\]
Consequently, \(\mathcal E\) contains no root-adjacent pairs.
\end{lemma}

\begin{proof}
Every vector in \(\mathcal E\) has all coordinates even.  Hence, for any
\(x,y\in\mathcal E\), every coordinate of \(x-y\) is even.  On the other hand,
a root \(e_a-e_b\) has one coordinate equal to \(1\), one coordinate equal to
\(-1\), and all other coordinates equal to \(0\).  Therefore \(x-y\) cannot
equal a root.
\end{proof}

It follows that every subset of \(\mathcal E\) is root-separated.  Moreover,
any two disjoint subsets of \(\mathcal E\) are mutually root-separated.  Thus,
once the codeword supports are chosen as disjoint subsets of \(\mathcal E\),
all root-error Knill--Laflamme conditions vanish automatically by
Proposition~\ref{prop:root-vanishing-weight}.

\subsection{Graph-theoretic interpretation}

The even-entry layer has a natural interpretation in terms of the complete
graph \(K_q\).  Indeed, an even-entry occupation vector of total weight \(4\)
is determined either by one label occupied with multiplicity \(4\), or by an
unordered pair of labels each occupied with multiplicity \(2\).  These two
possibilities are naturally identified with the vertices and edges of \(K_q\).

Let \(V(K_q)\) and \(E(K_q)\) denote the vertex set and edge set of the complete
graph on \(\{0,1,\dots,q-1\}\).  We identify
\[
    4e_i
    \quad\longleftrightarrow\quad
    i\in V(K_q),
\]
and
\[
    2e_i+2e_j
    \quad\longleftrightarrow\quad
    \{i,j\}\in E(K_q).
\]
Thus
\[
    \mathcal E \cong V(K_q)\sqcup E(K_q).
\]

This identification also gives a simple form for the Cartan expectations.
Let
\[
    H=\diag(h_0,\dots,h_{q-1}),
    \qquad
    \sum_{i=0}^{q-1}h_i=0,
\]
be an arbitrary element of the standard Cartan subalgebra.  We identify an
occupation vector with its associated weight functional.  Thus
\[
    (4e_i)(H)=4h_i,
\]
while
\[
    (2e_i+2e_j)(H)=2h_i+2h_j.
\]
Therefore, after root separation has eliminated the root-error conditions, the
remaining Cartan conditions become linear balancing conditions on weighted
collections of vertices and edges.  In the constructions below, the packets
are chosen so that each such weighted collection has Cartan expectation zero.

\subsection{Reformulation of the construction problem}

By Theorem~\ref{thm:support-separation}, to construct a
\(((4,q,2))_q\) permutation-invariant code it suffices to find \(q\)
orthonormal states
\[
    \ket{\psi_0},\dots,\ket{\psi_{q-1}}\in \Sym^4(\CC^q)
\]
whose supports
\[
    \mathcal S_0,\dots,\mathcal S_{q-1}
    \subseteq \mathcal E
\]
are pairwise disjoint and whose Cartan expectations are independent of the
logical basis state.

Since the supports lie in \(\mathcal E\), the root-error conditions are
automatic.  The remaining requirements are:

\begin{enumerate}
\item \emph{Disjointness:} the supports
\[
    \mathcal S_0,\dots,\mathcal S_{q-1}
\]
are pairwise disjoint subsets of \(\mathcal E\).

\item \emph{Normalization:} each state supported on \(\mathcal S_r\) has
squared coefficients
\[
    a_{\lambda}^{(r)}\ge 0,
    \qquad
    \sum_{\lambda\in\mathcal S_r} a_{\lambda}^{(r)}=1.
\]

\item \emph{Cartan balance:} for every traceless diagonal
\[
    H=\diag(h_0,\dots,h_{q-1}),
    \qquad
    \sum_i h_i=0,
\]
the quantity
\[
    \sum_{\lambda\in\mathcal S_r}
    \lambda(H)a_{\lambda}^{(r)}
\]
is independent of \(r\).  In the constructions below, this common value will
be \(0\).
\end{enumerate}

In graph-theoretic terms, the problem is to organize the vertices and edges of
\(K_q\) into \(q\) disjoint packets, and then assign
nonnegative weights to the elements of each packet so that the weighted
Cartan averages agree across all packets.  In the constructions below, this
common average is zero.

The preceding discussion suggests the following construction principle: choose the supports \(\mathcal S_r\) by decomposing the vertex-edge set
\(V(K_q)\sqcup E(K_q)\) into \(q\) disjoint packets with enough symmetry to
force Cartan balance. The two constructions in the next sections implement this principle in
different ways depending on the parity of \(q\).  For odd \(q\), the packets
come from the cyclic midpoint coloring of the edges of \(K_q\).  For even
\(q\), they come from a \(1\)-factorization of \(K_q\), together with a
separate vertex packet.

Figure~\ref{fig:K3-K4-packets} illustrates the packet decompositions for \(q=3,4,5,6\), highlighting the distinction between odd and even local dimensions. For odd \(q\), each packet consists of a vertex together with a matching covering all remaining vertices. For even \(q\), the edges partition into \(q-1\) perfect matchings, each forming a packet, while all vertices together form the final packet.

\begin{figure}[t]
\centering
\begin{minipage}[t]{0.48\linewidth}
\centering
\begin{tikzpicture}[
    scale=0.9,
    vertex/.style={circle,draw=black,thick,minimum size=7mm,
                   inner sep=0pt,font=\small\bfseries},
    edge/.style={line width=2.2pt}
]
    \path[use as bounding box] (-2.2,-2.1) rectangle (2.2,2.1);
    \coordinate (v0) at (0,1.8);
    \coordinate (v1) at (-1.6,-0.9);
    \coordinate (v2) at (1.6,-0.9);
    \draw[edge,RoyalBlue] (v1)--(v2);
    \draw[edge,OrangeRed] (v2)--(v0);
    \draw[edge,ForestGreen] (v0)--(v1);
    \node[vertex,fill=RoyalBlue,text=white] at (v0) {0};
    \node[vertex,fill=OrangeRed,text=white] at (v1) {1};
    \node[vertex,fill=ForestGreen,text=white] at (v2) {2};
\end{tikzpicture}
\par\smallskip
\small (a) Support packets for \(K_3\)
\end{minipage}
\hfill
\begin{minipage}[t]{0.48\linewidth}
\centering
\begin{tikzpicture}[
    scale=0.9,
    vertex/.style={circle,draw=black,thick,minimum size=7mm,
                   inner sep=0pt,font=\small\bfseries},
    edge/.style={line width=2.2pt}
]
    \path[use as bounding box] (-2.2,-2.1) rectangle (2.2,2.1);
    \coordinate (v0) at (90:1.8);
    \coordinate (v1) at (0:1.8);
    \coordinate (v2) at (-90:1.8);
    \coordinate (v3) at (-180:1.8);
    \draw[edge,RoyalBlue] (v0)--(v3);
    \draw[edge,RoyalBlue] (v1)--(v2);
    \draw[edge,OrangeRed] (v1)--(v3);
    \draw[edge,OrangeRed] (v2)--(v0);
    \draw[edge,ForestGreen] (v2)--(v3);
    \draw[edge,ForestGreen] (v0)--(v1);
    \node[vertex,fill=yellow!90,text=black] at (v0) {0};
    \node[vertex,fill=yellow!90,text=black] at (v1) {1};
    \node[vertex,fill=yellow!90,text=black] at (v2) {2};
    \node[vertex,fill=yellow!90,text=black] at (v3) {3};
\end{tikzpicture}
\par\smallskip
\small (b) Support packets for \(K_4\)
\end{minipage}

\par\vspace{8mm}

\begin{minipage}[t]{0.48\linewidth}
\centering
\begin{tikzpicture}[
    scale=0.9,
    vertex/.style={circle,draw=black,thick,minimum size=7mm,
                   inner sep=0pt,font=\small\bfseries},
    edge/.style={line width=2.2pt}
]
    \path[use as bounding box] (-2.2,-2.1) rectangle (2.2,2.1);
    \coordinate (v0) at (90:1.8);
    \coordinate (v1) at (18:1.8);
    \coordinate (v2) at (-54:1.8);
    \coordinate (v3) at (-126:1.8);
    \coordinate (v4) at (-198:1.8);
    \draw[edge,RoyalBlue] (v1)--(v4);
    \draw[edge,RoyalBlue] (v2)--(v3);
    \draw[edge,OrangeRed] (v2)--(v0);
    \draw[edge,OrangeRed] (v3)--(v4);
    \draw[edge,ForestGreen] (v3)--(v1);
    \draw[edge,ForestGreen] (v4)--(v0);
    \draw[edge,Purple] (v4)--(v2);
    \draw[edge,Purple] (v0)--(v1);
    \draw[edge,BurntOrange] (v0)--(v3);
    \draw[edge,BurntOrange] (v1)--(v2);
    \node[vertex,fill=RoyalBlue,text=white] at (v0) {0};
    \node[vertex,fill=OrangeRed,text=white] at (v1) {1};
    \node[vertex,fill=ForestGreen,text=white] at (v2) {2};
    \node[vertex,fill=Purple,text=white] at (v3) {3};
    \node[vertex,fill=BurntOrange,text=white] at (v4) {4};
\end{tikzpicture}
\par\smallskip
\small (c) Support packets for \(K_5\)
\end{minipage}
\hfill
\begin{minipage}[t]{0.48\linewidth}
\centering
\begin{tikzpicture}[
    scale=0.9,
    vertex/.style={circle,draw=black,thick,minimum size=7mm,
                   inner sep=0pt,font=\small\bfseries},
    edge/.style={line width=2.2pt}
]
    \path[use as bounding box] (-2.2,-2.1) rectangle (2.2,2.1);
    \coordinate (v0) at (90:1.8);
    \coordinate (v1) at (30:1.8);
    \coordinate (v2) at (-30:1.8);
    \coordinate (v3) at (-90:1.8);
    \coordinate (v4) at (-150:1.8);
    \coordinate (v5) at (-210:1.8);
    \draw[edge,RoyalBlue] (v0)--(v5);
    \draw[edge,RoyalBlue] (v1)--(v4);
    \draw[edge,RoyalBlue] (v2)--(v3);
    \draw[edge,OrangeRed] (v1)--(v5);
    \draw[edge,OrangeRed] (v2)--(v0);
    \draw[edge,OrangeRed] (v3)--(v4);
    \draw[edge,ForestGreen] (v2)--(v5);
    \draw[edge,ForestGreen] (v3)--(v1);
    \draw[edge,ForestGreen] (v4)--(v0);
    \draw[edge,Purple] (v3)--(v5);
    \draw[edge,Purple] (v4)--(v2);
    \draw[edge,Purple] (v0)--(v1);
    \draw[edge,BurntOrange] (v4)--(v5);
    \draw[edge,BurntOrange] (v0)--(v3);
    \draw[edge,BurntOrange] (v1)--(v2);
    \node[vertex,fill=yellow!90,text=black] at (v0) {0};
    \node[vertex,fill=yellow!90,text=black] at (v1) {1};
    \node[vertex,fill=yellow!90,text=black] at (v2) {2};
    \node[vertex,fill=yellow!90,text=black] at (v3) {3};
    \node[vertex,fill=yellow!90,text=black] at (v4) {4};
    \node[vertex,fill=yellow!90,text=black] at (v5) {5};
\end{tikzpicture}
\par\smallskip
\small (d) Support packets for \(K_6\)
\end{minipage}
\caption{The graph packets underlying the \(q=3,4,5,6\) constructions.
A vertex \(i\) represents the occupation vector \(4e_i\), and an edge
\(\{i,j\}\) represents \(2e_i+2e_j\). Elements of the same color belong
to the same support packet. For odd \(q\), each packet consists of a
vertex and the matching whose edges have that vertex as their midpoint
modulo \(q\). For even \(q\), each edge color gives a perfect matching,
and all yellow vertices form one additional packet.}
\label{fig:K3-K4-packets}
\end{figure}

\section{A motivating example for odd \texorpdfstring{\(q\)}{q}: the case \texorpdfstring{\(q=3\)}{q=3}}
\label{sec:su3-example}

We illustrate the construction in the first nontrivial case
\[
    V=\Sym^4(\CC^3),
\]
the irreducible \(\SU(3)\)-module of highest weight \((4,0)\).  The point of
the example is to show explicitly how the even-entry layer, root separation,
and Cartan balancing combine to produce a distance-two permutation-invariant
code.

We use the Cartan basis
\[
    H_1=E_{00}-E_{11},
    \qquad
    H_2=E_{11}-E_{22}.
\]
Thus an occupation state \(\ket{a_0,a_1,a_2}\) has weight
\[
    (a_0-a_1,\; a_1-a_2).
\]
For \(n=4\), the even-entry layer consists of occupations states $(4,0,0)$, $(0,4,0)$, $(0,0,4)$ and $(2,2,0)$, $(2,0,2)$, $(0,2,2)$. We can also write these as
\[
    4e_0,\;4e_1,\;4e_2
    \qquad\text{and}\qquad
    2e_0+2e_1,\;2e_0+2e_2,\;2e_1+2e_2.
\]
In the graph-theoretic picture of Section~\ref{sec:combinatorial-model},
these are the three vertices and three edges of \(K_3\).

\subsection{The support packets}

We choose the initial packet
\[
    \mathcal S_0
    =
    \{\,4e_0,\;2e_1+2e_2\,\}.
\]
Equivalently, in occupation notation,
\[
    \mathcal S_0
    =
    \{\, (4,0,0),\;(0,2,2)\,\}.
\]
Its cyclic translates are
\[
    \mathcal S_1
    =
    \{\, (0,4,0),\;(2,0,2)\,\},
    \qquad
    \mathcal S_2
    =
    \{\, (0,0,4),\;(2,2,0)\,\}.
\]
These three packets are pairwise disjoint subsets of the even-entry layer.
Therefore, by Lemma~\ref{lem:even-layer-root-sep}, each packet is
root-separated and the packets are mutually root-separated.  Hence all
root-error Knill--Laflamme conditions vanish automatically.

\subsection{Cartan balancing}

We now choose coefficients on \(\mathcal S_0\).  Write
\[
    \ket{\psi_0}
    =
    c_1\ket{4,0,0}
    +
    c_2\ket{0,2,2},
\]
where \(\ket{a_0,a_1,a_2}\) denotes the normalized occupation-number state.

The Cartan operator \(H_2\) has eigenvalue \(0\) on both occupation states:
\[
    H_2\ket{4,0,0}=0,
    \qquad
    H_2\ket{0,2,2}=0.
\]
For \(H_1\), the eigenvalues are
\[
    H_1\ket{4,0,0}=4\ket{4,0,0},
    \qquad
    H_1\ket{0,2,2}=-2\ket{0,2,2}.
\]
Therefore
\[
    \bra{\psi_0}H_1\ket{\psi_0}
    =
    4|c_1|^2-2|c_2|^2.
\]
The Cartan balance condition
\[
    \bra{\psi_0}H\ket{\psi_0}=0
    \qquad
    \text{for all }H\in\mathfrak h
\]
is therefore equivalent to
\[
    4|c_1|^2-2|c_2|^2=0,
    \qquad
    |c_1|^2+|c_2|^2=1.
\]
Thus
\[
    |c_1|^2=\frac13,
    \qquad
    |c_2|^2=\frac23.
\]
Choosing positive real coefficients gives
\[
    \ket{\psi_0}
    =
    \frac1{\sqrt3}\ket{4,0,0}
    +
    \sqrt{\frac23}\ket{0,2,2}.
\]

Now define
\[
    \ket{\psi_r}:=(X^{\otimes 4})^r\ket{\psi_0},
    \qquad
    r=0,1,2.
\]
Since \(X\) cyclically permutes the Cartan subalgebra, the same Cartan balance
holds for each cyclic translate.  Hence the cyclic support-separation
criterion, Corollary~\ref{cor:cyclic-separation}, applies.

For the logical basis in this example, we henceforth write
\(\ket{\mathbf r}:=\ket{\psi_r}\) for \(r=0,1,2\).

\subsection{The codewords}

Translating from normalized occupation states to symmetrized computational
words, we have
\[
    \ket{0,2,2}
    =
    \frac1{\sqrt6}\ket{\overline{1122}},
\]
where \(\ket{\overline{1122}}\) denotes the sum of the six distinct
permutations of the word \(1122\).  Thus
\[
    \sqrt{\frac23}\ket{0,2,2}
    =
    \frac13\ket{\overline{1122}}.
\]
Therefore the three logical basis states are
\begin{align*}
    \ket{\mathbf 0}
    &=
    \frac1{\sqrt3}\ket{0000}
    +
    \frac13\ket{\overline{1122}},\\
    \ket{\mathbf 1}
    &=
    \frac1{\sqrt3}\ket{1111}
    +
    \frac13\ket{\overline{0022}},\\
    \ket{\mathbf 2}
    &=
    \frac1{\sqrt3}\ket{2222}
    +
    \frac13\ket{\overline{0011}}.
\end{align*}

\begin{proposition}
The vectors
\[
    \ket{\mathbf 0},\ket{\mathbf 1},\ket{\mathbf 2}
\]
span a permutation-invariant quantum code in \(\Sym^4(\CC^3)\) with
parameters
\[
    ((4,3,2))_3.
\]
\end{proposition}

\begin{proof}
Each vector is normalized by construction, and the three vectors have
pairwise disjoint supports in the occupation-number basis.  Hence they are
orthonormal, so their span has dimension \(3\).

The support packets are pairwise disjoint subsets of the even-entry layer, so
all root-error matrix elements vanish by Lemma~\ref{lem:even-layer-root-sep}
and Proposition~\ref{prop:root-vanishing-weight}.  The coefficients were
chosen so that the Cartan expectation of \(\ket{\mathbf 0}\) is zero for every
\(H\in\mathfrak h\), and the cyclic shift transports this condition to
\(\ket{\mathbf 1}\) and \(\ket{\mathbf 2}\).  Therefore
Corollary~\ref{cor:cyclic-separation} implies that the span detects all
traceless one-body errors.  The code has distance at least \(2\), and hence
has parameters \(((4,3,2))_3\).
\end{proof}

This example is the odd-\(q\) construction in its smallest instance.  Each
logical basis vector is supported on one vertex of \(K_3\) and the opposite
edge, and the coefficient ratio is uniquely determined by Cartan balance.

The code also has a simple transversal Pauli action.  With
\(\omega=e^{2\pi i/3}\) and \(Z\ket a=\omega^a\ket a\), each occupation vector
appearing in \(\ket{\mathbf r}\) has total label sum congruent to \(r\) modulo
\(3\).  Hence
\[
    Z^{\otimes 4}\ket{\mathbf r}
    =
    \omega^r\ket{\mathbf r}.
\]
Together with
\[
    X^{\otimes 4}\ket{\mathbf r}=\ket{\mathbf {r+1}},
\]
this shows that the physical transversal operators \(X^{\otimes 4}\) and
\(Z^{\otimes 4}\) implement the standard logical qutrit Pauli operators.

This \(q=3\) code reproduces the \((4,0)\) \(\SU(3)\)-representation code of
Herbert, Gross, and Newman
\cite{herbert2023qutritcodesrepresentationssu3}.  That construction was the
starting point for the all-\(q\) family developed here.

\section{The general odd-\texorpdfstring{$q$}{q} construction}
\label{sec:general-construction}

We now construct permutation-invariant codes in
\[
    V=\Sym^4(\CC^q)
\]
for every odd integer \(q\ge 3\).  In the language of
Section~\ref{sec:combinatorial-model}, the construction decomposes the
vertex-edge set
\[
    V(K_q)\sqcup E(K_q)
\]
into \(q\) disjoint packets satisfying the Cartan balancing conditions.  For
odd \(q\), this decomposition is obtained from the cyclic midpoint coloring of
the edges of \(K_q\).

\subsection{Midpoint packets}

Because \(q\) is odd, the element \(2\) is invertible modulo \(q\).  Hence
every edge \(\{i,j\}\) of \(K_q\) has a unique midpoint
\[
    r \equiv \frac{i+j}{2} \pmod q.
\]
For each \(r\in\{0,1,\dots,q-1\}\), let \(M_r\) be the set of edges whose
midpoint is \(r\):
\[
    M_r
    :=
    \bigl\{\{i,j\}\in E(K_q): i+j\equiv 2r \pmod q\bigr\}.
\]
Equivalently,
\[
    M_r
    =
    \bigl\{\{r+i,r-i\}:1\le i\le \tfrac{q-1}{2}\bigr\},
\]
with indices understood modulo \(q\).  Each \(M_r\) is a near-perfect matching: its edges are pairwise
disjoint, it contains \((q-1)/2\) edges, and it covers every vertex except
\(r\).

We define the \(r\)-th support packet by adjoining this omitted vertex:
\[
    \mathcal S_r
    :=
    \{\,4e_r\,\}
    \cup
    \{\,2e_i+2e_j:\{i,j\}\in M_r\,\}.
\]
Equivalently,
\[
    \mathcal S_r
    =
    \{\,4e_r\,\}
    \cup
    \left\{
        2e_{r+i}+2e_{r-i}
        :
        1\le i\le \frac{q-1}{2}
    \right\}.
\]

\subsection{The associated states}

Let
\[
    \ket{a_0,\dots,a_{q-1}}
\]
denote the normalized occupation-number state in \(\Sym^4(\CC^q)\).  Thus
\[
    \ket{4e_r}
\]
is the pure word \(\ket{rrrr}\), while
\[
    \ket{2e_i+2e_j}
\]
is the normalized symmetric state with two copies of \(i\) and two copies of
\(j\).

We define
\[
    \ket{\psi_0}
    =
    c_0\ket{4e_0}
    +
    \sum_{i=1}^{(q-1)/2}
    c_i\ket{2e_i+2e_{q-i}},
\]
and then set
\[
    \ket{\psi_r}:=(X^{\otimes 4})^r\ket{\psi_0},
    \qquad
    r=0,1,\dots,q-1.
\]
By construction,
\[
    \supp\ket{\psi_r}=\mathcal S_r.
\]

\subsection{Disjointness and root separation}

\begin{lemma}
\label{lem:general-root-separation}
For odd \(q\), the supports
\[
    \mathcal S_0,\mathcal S_1,\dots,\mathcal S_{q-1}
\]
are pairwise disjoint.  Moreover, each \(\mathcal S_r\) is root-separated,
and the supports are mutually root-separated.
\end{lemma}

\begin{proof}
We first prove disjointness.  Clearly
\[
    4e_r=4e_s
    \quad\Longrightarrow\quad
    r=s.
\]
Also, \(4e_r\) cannot equal \(2e_i+2e_j\), since the former has one nonzero
coordinate equal to \(4\), while the latter has two nonzero coordinates equal
to \(2\).

Now suppose
\[
    2e_i+2e_j\in \mathcal S_r\cap \mathcal S_s.
\]
Then the edge \(\{i,j\}\) lies in both \(M_r\) and \(M_s\).  Hence
\[
    i+j\equiv 2r\equiv 2s\pmod q.
\]
Since \(2\) is invertible modulo \(q\), we get
\[
    r\equiv s\pmod q,
\]
and hence \(r=s\).  Therefore the supports are pairwise disjoint.

Each \(\mathcal S_r\) is contained in the even-entry layer
\(\mathcal E\).  By Lemma~\ref{lem:even-layer-root-sep}, every subset of
\(\mathcal E\) is root-separated.  Since the supports are disjoint subsets of
\(\mathcal E\), they are also mutually root-separated.
\end{proof}

Thus all root-error Knill--Laflamme conditions vanish automatically.

\subsection{Cartan balancing}

It remains to impose the Cartan conditions.  We use the Cartan basis
\[
    H_j=E_{j-1,j-1}-E_{j,j},
    \qquad
    1\le j\le q-1.
\]
Equivalently, if
\[
    H=\diag(h_0,\dots,h_{q-1}),
    \qquad
    \sum_{a=0}^{q-1}h_a=0,
\]
then an occupation vector \(a=(a_0,\dots,a_{q-1})\) contributes
\[
    a(H)=\sum_{\ell=0}^{q-1} a_\ell h_\ell.
\]

For
\[
    \ket{\psi_0}
    =
    c_0\ket{4e_0}
    +
    \sum_{i=1}^{(q-1)/2}
    c_i\ket{2e_i+2e_{q-i}},
\]
the Cartan expectation is
\[
    \bra{\psi_0}H\ket{\psi_0}
    =
    4|c_0|^2 h_0
    +
    \sum_{i=1}^{(q-1)/2}
    2|c_i|^2(h_i+h_{q-i}).
\]
We want this to vanish for every traceless diagonal \(H\).  Since
\[
    h_0+\sum_{i=1}^{(q-1)/2}(h_i+h_{q-i})=0,
\]
this is achieved precisely when the coefficients of all grouped variables
\[
    h_0,\quad h_i+h_{q-i}
\]
are equal.  Thus the Cartan balance equations are
\[
    4|c_0|^2=2|c_i|^2
    \qquad
    \left(1\le i\le \frac{q-1}{2}\right),
\]
together with normalization
\[
    |c_0|^2+\sum_{i=1}^{(q-1)/2}|c_i|^2=1.
\]
Equivalently,
\[
    |c_i|^2=2|c_0|^2
    \qquad
    \left(1\le i\le \frac{q-1}{2}\right).
\]
Substituting into the normalization gives
\[
    |c_0|^2+\frac{q-1}{2}\cdot 2|c_0|^2
    =
    q|c_0|^2
    =
    1.
\]
Hence
\[
    |c_0|^2=\frac1q,
    \qquad
    |c_i|^2=\frac2q
    \quad
    \left(1\le i\le \frac{q-1}{2}\right).
\]
Choosing positive real coefficients gives
\[
    c_0=\frac1{\sqrt q},
    \qquad
    c_i=\sqrt{\frac2q}.
\]

\subsection{The resulting code}

With these coefficients, we adopt the logical notation
\(\ket{\mathbf r}:=\ket{\psi_r}\).  In particular,
\[
    \ket{\mathbf 0}
    =
    \frac1{\sqrt q}\ket{4e_0}
    +
    \sum_{i=1}^{(q-1)/2}
    \sqrt{\frac2q}\,
    \ket{2e_i+2e_{q-i}}.
\]
Equivalently, in symmetrized word notation,
\[
    \ket{\mathbf 0}
    =
    \frac1{\sqrt q}\ket{0000}
    +
    \frac1{\sqrt{3q}}
    \sum_{i=1}^{(q-1)/2}
    \ket{\overline{i,i,q-i,q-i}},
\]
where \(\ket{\overline{i,i,q-i,q-i}}\) denotes the sum of the six distinct
permutations of the word \(i,i,q-i,q-i\).  Accordingly,
\[
    \ket{\mathbf r}:=(X^{\otimes 4})^r\ket{\mathbf 0},
    \qquad
    r=0,1,\dots,q-1.
\]

\begin{theorem}
\label{thm:general-family}
For every odd integer \(q\ge 3\), the vectors
\[
    \ket{\mathbf 0},\ket{\mathbf 1},\dots,\ket{\mathbf {q-1}}
\]
span a permutation-invariant quantum code in \(\Sym^4(\CC^q)\) with
parameters
\[
    ((4,q,2))_q.
\]
\end{theorem}

\begin{proof}
The vectors are normalized by construction.  By
Lemma~\ref{lem:general-root-separation}, their supports are pairwise
disjoint, so the vectors are mutually orthogonal.  Hence their span has
dimension \(q\).

The same lemma shows that the supports are root-separated and mutually
root-separated.  Therefore all root-error matrix elements vanish.  The
coefficients were chosen so that
\[
    \bra{\mathbf 0}H\ket{\mathbf 0}=0
    \qquad
    \text{for all }H\in\mathfrak h.
\]
Since the remaining codewords are cyclic translates of \(\ket{\mathbf 0}\),
Corollary~\ref{cor:cyclic-separation} implies that the span detects all
traceless one-body errors.  Thus the code has distance at least \(2\).
Since its length is \(4\) and its dimension is \(q\), it has parameters
\[
    ((4,q,2))_q.
\]
\end{proof}

\begin{proposition}[Transversal Pauli action for odd \(q\)]
\label{prop:odd-transversal-pauli}
Let \(q\ge 3\) be odd, and let \(\mathcal C\subseteq \Sym^4(\CC^q)\) be the
code of Theorem~\ref{thm:general-family}.  Let
\[
    \begin{gathered}
    X\ket a=\ket{a+1\!\!\!\pmod q},\qquad Z\ket a=\omega^a\ket a,\\
    \omega=e^{2\pi i/q}.
    \end{gathered}
\]
Then \(\mathcal C\) is preserved by \(X^{\otimes 4}\) and \(Z^{\otimes 4}\).
In the logical basis \(\{\ket{\mathbf r}\}_{r\in \Z_q}\),
\[
    X^{\otimes 4}\ket{\mathbf r}
    =
    \ket{\mathbf {r+1}},
    \qquad
    Z^{\otimes 4}\ket{\mathbf r}
    =
    \omega^{4r}\ket{\mathbf r}.
\]
Equivalently,
\[
    X^{\otimes 4}=X_L,
    \qquad
    Z^{\otimes 4}=Z_L^{\,4}.
\]
Since \(q\) is odd, 4 is invertible mod $ q $, so \(Z_L^{\,4}\) generates
the full logical phase subgroup.  Thus the odd-\(q\) family admits a
transversal realization of the full logical Pauli group.
\end{proposition}

\begin{proof}
The identity
\[
    X^{\otimes 4}\ket{\mathbf r}=\ket{\mathbf {r+1}}
\]
holds by definition of the cyclic logical basis.  For the \(Z\)-action, recall
that
\[
    \supp\ket{\mathbf r}
    =
    \{4e_r\}
    \cup
    \left\{
        2e_{r+i}+2e_{r-i}
        :
        1\le i\le \frac{q-1}{2}
    \right\}.
\]
On an occupation vector \(a=(a_0,\dots,a_{q-1})\), the operator
\(Z^{\otimes 4}\) acts by the phase
\[
    \omega^{\sum_{j=0}^{q-1} j a_j}.
\]
For the vertex \(4e_r\), this exponent is \(4r\).  For an edge vector
\(2e_{r+i}+2e_{r-i}\), it is
\[
    2(r+i)+2(r-i)=4r
    \qquad \pmod q.
\]
Thus every occupation vector in \(\supp\ket{\mathbf r}\) acquires the same
phase \(\omega^{4r}\).  Therefore
\[
    Z^{\otimes 4}\ket{\mathbf r}
    =
    \omega^{4r}\ket{\mathbf r}.
\]
The final assertion follows because multiplication by \(4\) is invertible in
\(\Z_q\) when \(q\) is odd.
\end{proof}

\section{A motivating example for even \texorpdfstring{\(q\)}{q}: the case \texorpdfstring{\(q=4\)}{q=4}}
\label{sec:q4-example}

We now illustrate the even-\(q\) construction in the first case where the
perfect matching structure is nontrivial, namely
\[
    V=\Sym^4(\CC^4).
\]
The point of the example is to show how the support-separation criterion is
implemented when the edges of the complete graph are partitioned into perfect
matchings.  Equivalently, this is a \(1\)-factorization of \(K_4\).
For \(q=2\), the same construction gives the Hagiwara--Nakayama
four-qubit single-deletion code
\cite{Hagiwara2020FourQubitDeletion}: the vertex state is the GHZ state and
the edge state is the weight-two Dicke state.  The matching decomposition has
no further combinatorial content in this degenerate case.

\subsection{The even-entry layer and \texorpdfstring{\(K_4\)}{K4}}

For \(n=4\), the even-entry layer consists of the occupation vectors
\[
    4e_0,\,4e_1,\,4e_2,\,4e_3
    \qquad\text{and}\qquad
    2e_i+2e_j
    \quad
    (0\le i<j\le 3).
\]
In the graph-theoretic picture of Section~\ref{sec:combinatorial-model},
these correspond respectively to the vertices and edges of \(K_4\).

We apply the round-robin \(1\)-factorization of \(K_4\) with fixed point
\(3\), treating the remaining vertices as \(\mathbb Z_3=\{0,1,2\}\).  The
three perfect matchings are
\[
    \begin{aligned}
    F_0&=\{\{0,3\},\{1,2\}\},\\
    F_1&=\{\{1,3\},\{2,0\}\},\\
    F_2&=\{\{2,3\},\{0,1\}\}.
    \end{aligned}
\]
Each perfect matching covers every vertex exactly once.

\subsection{The support packets}

We define four support packets.  For each \(r\in\mathbb Z_3\), the edge packet is
\[
    \mathcal S_r
    :=
    \{\,2e_i+2e_j:\{i,j\}\in F_r\,\},
    \qquad
    r=0,1,2.
\]
Explicitly,
\begin{align*}
    \mathcal S_0&=\{\,2e_0+2e_3,\,2e_1+2e_2\,\},\\
    \mathcal S_1&=\{\,2e_1+2e_3,\,2e_2+2e_0\,\},\\
    \mathcal S_2&=\{\,2e_2+2e_3,\,2e_0+2e_1\,\}.
\end{align*}
The vertex packet is
\[
    \mathcal S_3
    :=
    \{\,4e_0,\,4e_1,\,4e_2,\,4e_3\,\}.
\]
These four packets are pairwise disjoint and together partition the
even-entry layer.

Since each \(\mathcal S_r\) is contained in the even-entry layer, each packet
is root-separated by Lemma~\ref{lem:even-layer-root-sep}.  Since the packets
are pairwise disjoint subsets of the even-entry layer, they are also mutually
root-separated.  Therefore all root-error Knill--Laflamme conditions vanish
automatically.

\subsection{The associated states}

The three matching states are
\begin{align*}
    \ket{\psi_0}
    &=
    \frac{1}{\sqrt 2}
    \left(
        \ket{2e_0+2e_3}
        +
        \ket{2e_1+2e_2}
    \right),\\
    \ket{\psi_1}
    &=
    \frac{1}{\sqrt 2}
    \left(
        \ket{2e_1+2e_3}
        +
        \ket{2e_2+2e_0}
    \right),\\
    \ket{\psi_2}
    &=
    \frac{1}{\sqrt 2}
    \left(
        \ket{2e_2+2e_3}
        +
        \ket{2e_0+2e_1}
    \right).
\end{align*}
The vertex state is
\[
    \ket{\psi_3}
    =
    \frac{1}{2}
    \left(
        \ket{4e_0}
        +\ket{4e_1}
        +\ket{4e_2}
        +\ket{4e_3}
    \right).
\]
Thus each \(\ket{\psi_r}\) is the uniform superposition, in the normalized
occupation-number basis, over the support packet \(\mathcal S_r\).

Equivalently, in the computational basis,
\begin{align*}
    \ket{\psi_0}
    &=
    \frac{1}{\sqrt{12}}
    \left(
        \ket{\overline{0033}}
        +
        \ket{\overline{1122}}
    \right),\\
    \ket{\psi_1}
    &=
    \frac{1}{\sqrt{12}}
    \left(
        \ket{\overline{1133}}
        +
        \ket{\overline{0022}}
    \right),\\
    \ket{\psi_2}
    &=
    \frac{1}{\sqrt{12}}
    \left(
        \ket{\overline{2233}}
        +
        \ket{\overline{0011}}
    \right),
\end{align*}
and
\[
    \ket{\psi_3}
    =
    \frac{1}{2}
    \left(
        \ket{0000}
        +\ket{1111}
        +\ket{2222}
        +\ket{3333}
    \right).
\]

\subsection{Cartan balancing}

Let
\[
    H=\diag(h_0,h_1,h_2,h_3),
    \qquad
    h_0+h_1+h_2+h_3=0,
\]
be an arbitrary Cartan element.

For the first matching state,
\begin{align*}
\bra{\psi_0}H\ket{\psi_0}
&=
\frac{1}{2}
\left[
    (2h_0+2h_3)
    +
    (2h_1+2h_2)
\right]\\
&=
h_0+h_1+h_2+h_3\\
&=
0.
\end{align*}
The same calculation applies to \(\ket{\psi_1}\) and \(\ket{\psi_2}\),
because each perfect matching covers every vertex exactly once.  For the
vertex state,
\[
\bra{\psi_3}H\ket{\psi_3}
=
\frac{1}{4}\sum_{i=0}^3 4h_i
=
\sum_{i=0}^3 h_i
=
0.
\]
Therefore
\[
    \bra{\psi_r}H\ket{\psi_r}=0
    \qquad
    \text{for all } r\in\{0,1,2,3\}
    \text{ and all } H\in\mathfrak h.
\]

\subsection{The resulting code}

\begin{proposition}
The vectors
\[
    \ket{\psi_0},\ket{\psi_1},\ket{\psi_2},\ket{\psi_3}
\]
span a permutation-invariant quantum code in \(\Sym^4(\CC^4)\) with
parameters
\[
    ((4,4,2))_4.
\]
\end{proposition}

\begin{proof}
Each \(\ket{\psi_r}\) is normalized by construction.  Since the supports
\(\mathcal S_0,\mathcal S_1,\mathcal S_2,\mathcal S_3\) are pairwise
disjoint in the occupation-number basis, the four vectors are mutually
orthogonal.  Hence their span has dimension \(4\).

The support packets are pairwise disjoint subsets of the even-entry layer.
Thus they are root-separated and mutually root-separated by
Lemma~\ref{lem:even-layer-root-sep}.  The Cartan expectations vanish by the
calculation above.  Therefore the hypotheses of
Theorem~\ref{thm:support-separation} are satisfied, and the span detects all
traceless one-body errors.  Hence the code has distance at least \(2\), and
therefore has parameters
\[
    ((4,4,2))_4.
\]
\end{proof}

This example is the smallest even-\(q\) instance in which the graph-theoretic
mechanism is visible.  The case \(q=2\) recovers the previously known
Hagiwara--Nakayama code: \(K_2\) has one edge, so the construction consists of
one vertex packet and one edge packet.  For \(q=4\), the supports are obtained by decomposing the edge set of
\(K_4\) into perfect matchings, together with a separate vertex packet.  In
the general even-\(q\) construction, the same mechanism applies using a
\(1\)-factorization of \(K_q\).

\section{The general even-\texorpdfstring{\(q\)}{q} construction}
\label{sec:even-construction}

We now construct permutation-invariant codes in
\[
    V=\Sym^4(\CC^q)
\]
for every even integer \(q\ge 2\).  In the language of
Section~\ref{sec:combinatorial-model}, the construction decomposes the
vertex-edge set
\[
    V(K_q)\sqcup E(K_q)
\]
into \(q\) disjoint packets satisfying the Cartan balancing conditions.  For
even \(q\), this is accomplished by decomposing the edge set of \(K_q\) into
perfect matchings, that is, by choosing a \(1\)-factorization of \(K_q\).

\subsection{Perfect matchings}

We use the standard round-robin \(1\)-factorization of \(K_q\).  Treat the
vertex set as
\[
    \mathbb Z_{q-1}\sqcup\{q-1\},
\]
where \(q-1\) plays the role of a fixed point.  For each
\(r\in\mathbb Z_{q-1}\), define
\[
    F_r
    :=
    \bigl\{\{r,\,q-1\}\bigr\}
    \cup
    \left\{
        \{r+i,\,r-i\}:1\le i\le \frac{q}{2}-1
    \right\},
\]
where arithmetic in the second set is taken modulo \(q-1\).  Each \(F_r\) is
a perfect matching: it contains \(q/2\) disjoint edges and covers every
vertex exactly once.  Moreover,
\[
    E(K_q)=\bigsqcup_{r\in\mathbb Z_{q-1}}F_r,
\]
so the \(F_r\) form a \(1\)-factorization of \(K_q\).

For \(r\in\mathbb Z_{q-1}\), define the edge packet
\[
    \mathcal S_r
    :=
    \{\,2e_i+2e_j:\{i,j\}\in F_r\,\}.
\]
Equivalently,
\[
    \mathcal S_r
    =
    \{\,2e_r+2e_{q-1}\,\}
    \cup
    \left\{
        2e_{r+i}+2e_{r-i}:1\le i\le \frac{q}{2}-1
    \right\}.
\]

We also define the vertex packet
\[
    \mathcal S_{q-1}
    :=
    \{\,4e_i:0\le i\le q-1\,\}.
\]
Then
\[
    \mathcal S_0,\mathcal S_1,\dots,\mathcal S_{q-2},\mathcal S_{q-1}
\]
partition the full even-entry layer
\[
    \mathcal E\cong V(K_q)\sqcup E(K_q).
\]

\subsection{The associated states}

Let
\[
    \ket{a_0,\dots,a_{q-1}}
\]
denote the normalized occupation-number state in \(\Sym^4(\CC^q)\).  Thus
\(\ket{4e_i}\) is the pure word \(\ket{i i i i}\), while for \(i\neq j\),
\[
    \ket{2e_i+2e_j}
    =
    \frac{1}{\sqrt 6}\ket{\overline{i i j j}},
\]
where \(\ket{\overline{i i j j}}\) denotes the sum of the six distinct
permutations of the word \(i i j j\).

For \(r\in\mathbb Z_{q-1}\), define the matching states
\[
    \ket{\psi_r}
    :=
    \sqrt{\frac{2}{q}}
    \left(
        \ket{2e_r+2e_{q-1}}
        +
        \sum_{i=1}^{q/2-1}
        \ket{2e_{r+i}+2e_{r-i}}
    \right),
\]
with arithmetic in \(\mathbb Z_{q-1}\).  Since each \(F_r\) contains
\(q/2\) edges, the coefficient \(\sqrt{2/q}\) normalizes \(\ket{\psi_r}\).

Define the vertex state
\[
    \ket{\psi_{q-1}}
    :=
    \frac{1}{\sqrt q}
    \sum_{i=0}^{q-1}
    \ket{4e_i}.
\]

Equivalently, each codeword is the uniform superposition, in the normalized
occupation-number basis, over its support packet:
\[
    \ket{\psi_r}
    =
    \sqrt{\frac{2}{q}}
    \sum_{a\in\mathcal S_r}\ket a
    \; (r\in\mathbb Z_{q-1}),
    \quad
    \ket{\psi_{q-1}}
    =
    \frac{1}{\sqrt{q}} 
    \sum_{a\in\mathcal S_{q-1}}\ket a.
\]
By construction,
\[
    \supp\ket{\psi_r}=\mathcal S_r
    \quad(r\in\mathbb Z_{q-1}),
    \qquad
    \supp\ket{\psi_{q-1}}=\mathcal S_{q-1}.
\]

\subsection{Disjointness and root separation}

\begin{lemma}
\label{lem:even-root-separation}
The supports
\[
    \mathcal S_0,\mathcal S_1,\dots,\mathcal S_{q-2},\mathcal S_{q-1}
\]
are pairwise disjoint.  Moreover, each support is root-separated, and any two
distinct supports are mutually root-separated.
\end{lemma}

\begin{proof}
The packet \(\mathcal S_{q-1}\) consists of the vertex vectors \(4e_i\), while
each \(\mathcal S_r\), \(r\in\mathbb Z_{q-1}\), consists of the edge vectors
\(2e_i+2e_j\) with \(\{i,j\}\in F_r\).  A vertex vector cannot equal an edge
vector.  Since
\[
    E(K_q)=\bigsqcup_{r\in\mathbb Z_{q-1}}F_r,
\]
each edge vector \(2e_i+2e_j\) lies in exactly one edge packet.  Therefore
the packets are pairwise disjoint.

Each support is contained in the even-entry layer \(\mathcal E\).  Hence each
support is root-separated by Lemma~\ref{lem:even-layer-root-sep}.  Since the
supports are pairwise disjoint subsets of \(\mathcal E\), they are also
mutually root-separated.
\end{proof}

Thus all root-error Knill--Laflamme conditions vanish automatically.

\subsection{Cartan balancing}

It remains to verify the Cartan conditions.  Let
\[
    H=\diag(h_0,\dots,h_{q-1}),
    \qquad
    \sum_{i=0}^{q-1}h_i=0,
\]
be an arbitrary traceless diagonal operator.

For each \(r\in\mathbb Z_{q-1}\), since \(F_r\) is a perfect matching, every
vertex of \(K_q\) appears in exactly one edge of \(F_r\).  Therefore
\[
\bra{\psi_r}H\ket{\psi_r}
=
\frac{2}{q}
\sum_{\{i,j\}\in F_r}(2h_i+2h_j)
=
\frac{4}{q}
\sum_{i=0}^{q-1}h_i
=
0.
\]
For the vertex state,
\[
\bra{\psi_{q-1}}H\ket{\psi_{q-1}}
=
\frac{1}{q}
\sum_{i=0}^{q-1}4h_i
=
\frac{4}{q}
\sum_{i=0}^{q-1}h_i
=
0.
\]
Therefore
\[
    \bra{\psi_r}H\ket{\psi_r}=0
    \qquad
    \text{for all } r=0,1,\dots,q-1
    \text{ and all } H\in\mathfrak h.
\]

\subsection{The resulting code}

\begin{theorem}
\label{thm:even-family}
For every even integer \(q\ge 2\), the vectors
\[
    \ket{\psi_0},\ket{\psi_1},\dots,\ket{\psi_{q-2}},\ket{\psi_{q-1}}
\]
span a permutation-invariant quantum code in \(\Sym^4(\CC^q)\) with
parameters
\[
    ((4,q,2))_q.
\]
\end{theorem}

\begin{proof}
The vectors are normalized by construction.  By
Lemma~\ref{lem:even-root-separation}, their supports are pairwise disjoint,
so the vectors are mutually orthogonal.  Hence their span has dimension \(q\).

The same lemma shows that the supports are root-separated and mutually
root-separated.  Therefore all root-error matrix elements vanish.  The Cartan
expectations vanish by the computation above.  Thus the hypotheses of
Theorem~\ref{thm:support-separation} are satisfied, and the span detects all
traceless one-body errors.  Hence the code has distance at least \(2\).
Since its length is \(4\) and its dimension is \(q\), it has parameters
\[
    ((4,q,2))_q.
\]
\end{proof}

\section{Minimum block length}
\label{sec:n-le-3}

We now prove that no permutation-invariant code of dimension \(q\) in
\(\Sym^n(\CC^q)\) can have distance at least \(2\) when \(n\le3\). The 
argument we present is self-contained, for an independent second proof that uses linear programming bounds, see the appendix.  It uses the standard reference-system
characterization of exact quantum error correction---erasure correctability
implies decoupling \cite{SchumacherNielsen1996}---whose short proof we include
for completeness.

Let \(W:\CC^q\to(\CC^q)^{\otimes n}\) be an encoding isometry with
orthonormal codewords \(\ket{\psi_r}=W\ket r\), and introduce a
\(q\)-dimensional reference system \(R\).  The encoded maximally entangled
state is
\begin{equation}
\label{eq:encoded-Choi-state}
    \ket{\Omega}_{R1\cdots n}
    :=
    \frac1{\sqrt q}\sum_{r=0}^{q-1}
    \ket r_R\ket{\psi_r}_{1\cdots n}.
\end{equation}

\begin{lemma}[Erasure correction implies decoupling]
\label{lem:erasure-decoupling}
If erasure of a physical subsystem \(X\) is correctable on the code, then
\begin{equation}
\label{eq:decoupling}
    \rho_{RX}=\rho_R\otimes\rho_X
\end{equation}
for the state in \eqref{eq:encoded-Choi-state}.
\end{lemma}

\begin{proof}
Let \(X^c\) denote all physical subsystems other than \(X\), and set
\[
    \tau^X_{rs}
    :=
    \Tr_{X^c}\bigl(\ket{\psi_r}\!\bra{\psi_s}\bigr).
\]
Correctability of the erasure of \(X\) is equivalent to the
Knill--Laflamme conditions for every operator \(A_X\in\L(X)\):
\[
    \bra{\psi_s}A_X\ket{\psi_r}=c_A\delta_{rs}.
\]
Since
\[
    \bra{\psi_s}A_X\ket{\psi_r}
    =\Tr\!\left(A_X\tau^X_{rs}\right),
\]
and this holds for every \(A_X\), there is a density operator \(\sigma_X\),
independent of \(r\), such that
\[
    \tau^X_{rs}=\delta_{rs}\sigma_X.
\]
Taking the \(RX\) marginal of \eqref{eq:encoded-Choi-state} therefore gives
\[
    \rho_{RX}
    =\frac1q\sum_{r,s}\ket r\!\bra s_R\otimes\tau^X_{rs}
    =\frac{I_R}{q}\otimes\sigma_X
    =\rho_R\otimes\rho_X.
\]
\end{proof}

\begin{theorem}
\label{thm:no-small-length}
For every \(q\ge2\), there is no permutation-invariant code of dimension
\(q\) in \(\Sym^n(\CC^q)\) that detects all single-site errors when
\(n\le3\).  Equivalently, no permutation-invariant code with parameters
\[
    ((n,q,d))_q
\]
and \(d\ge2\) exists for \(n\le3\).
\end{theorem}

\begin{proof}
A distance-two code corrects the erasure of any one physical subsystem.

For \(n=1\), a \(q\)-dimensional code equals the entire physical Hilbert
space.  A nonscalar one-site operator cannot act as a scalar on that space,
so the error-detection conditions fail.

For \(n=2\), write the pure state \eqref{eq:encoded-Choi-state} as
\(\ket\Omega_{RAB}\), and set
\[
    r_A:=\rank\rho_A,
    \qquad
    r_B:=\rank\rho_B.
\]
Correctability of the erasure of \(A\) gives
\[
    \rho_{RA}=\frac{I_R}{q}\otimes\rho_A,
    \qquad
    \rank\rho_{RA}=q r_A.
\]
Because \(\ket\Omega_{RAB}\) is pure, the reduced states on the two sides of
the bipartition \(RA|B\) have the same rank.  Hence
\[
    r_B=q r_A.
\]
Since \(r_B\le q\) and \(r_A\ge1\), we must have \(r_A=1\) and \(r_B=q\).
Correctability of the erasure of \(B\), followed by purity across \(RB|A\),
similarly gives
\[
    r_A=q r_B=q^2,
\]
which is a contradiction.

It remains to exclude \(n=3\).  Write
\(\ket\Omega_{RABC}\) for the corresponding pure state, and set
\[
    r_A:=\rank\rho_A,
    \qquad
    r_B:=\rank\rho_B,
    \qquad
    r_C:=\rank\rho_C.
\]
Correctability of the erasure of \(A\), together with purity across
\(RA|BC\), gives
\[
    \rank\rho_{BC}=\rank\rho_{RA}=q r_A.
\]
The support of \(\rho_{BC}\) is contained in
\(\supp(\rho_B)\otimes\supp(\rho_C)\), so
\[
    q r_A\le r_Br_C.
\]
Cycling \(A,B,C\) yields
\[
    q r_A\le r_Br_C,
    \qquad
    q r_B\le r_Ar_C,
    \qquad
    q r_C\le r_Ar_B.
\]
Multiplying these inequalities and cancelling the positive quantity
\(r_Ar_Br_C\) gives
\[
    q^3\le r_Ar_Br_C.
\]
But \(r_A,r_B,r_C\le q\), so equality is forced and
\[
    r_A=r_B=r_C=q.
\]
Finally, decoupling for \(C\) and purity across \(RC|AB\) give
\[
    \rank\rho_{AB}=\rank\rho_{RC}=q r_C=q^2.
\]
Thus \(\rho_{AB}\) has full support on
\(\CC^q\otimes\CC^q\).

Permutation invariance now gives the contradiction.  Every codeword lies in
\(\Sym^3(\CC^q)\), hence is a \(+1\) eigenvector of the transposition of the
\(A\) and \(B\) factors.  Therefore \(\ket\Omega_{RABC}\) is a \(+1\)
eigenvector of \(\mathsf{SWAP}_{AB}\otimes I_{RC}\).  The \(+1\) eigenspace
of this operator is
\(\Sym^2(\CC^q)_{AB}\otimes\mathcal H_{RC}\), so
\(\ket\Omega_{RABC}\) lies in that subspace and \(\rho_{AB}\) is supported on
\(\Sym^2(\CC^q)\), whose dimension is
\[
    \dim\Sym^2(\CC^q)=\frac{q(q+1)}2<q^2
    \qquad(q\ge2).
\]
This is incompatible with the full support forced above.  Hence no such code
exists for \(n=3\).
\end{proof}

\begin{remark}
The rank argument isolates exactly where permutation invariance enters: error
correction forces a two-body marginal to have full rank, while symmetry
confines its support to the symmetric subspace.  More generally, a
\(((3,q,2))_q\) code saturates the quantum Singleton bound and its encoded
maximally entangled state is an absolutely maximally entangled state on four
\(q\)-dimensional systems.
For \(q\ge3\), this obstruction is genuinely imposed by permutation
invariance: unrestricted \(((3,q,2))_q\) codes exist, obtained by taking one
subsystem of an \(\operatorname{AME}(4,q)\) state as the reference.  For
\(q\neq6\), such states can be obtained from pairs of orthogonal Latin squares
\cite{BoseShrikhandeParker1960}; the exceptional \(q=6\) case follows from the
nonclassical \(\operatorname{AME}(4,6)\) construction of Rather et al.
\cite{RatherEtAlAME46}.  Thus permutation invariance raises the minimum block
length from three to four for every \(q\ge3\).
\end{remark}

As an immediate consequence, the length-four constructions are optimal.

\begin{corollary}
\label{cor:minimum-block-length}
For every integer \(q\ge2\), the code constructed in
Theorem~\ref{thm:general-family} when \(q\) is odd and in
Theorem~\ref{thm:even-family} when \(q\) is even has minimum block length
among permutation-invariant codes of dimension \(q\) and distance at least
\(2\).  Equivalently, four physical qudits are necessary and sufficient for a
distance-two permutation-invariant encoding of one logical qudit of local
dimension \(q\).
\end{corollary}

\begin{proof}
Theorems~\ref{thm:general-family} and~\ref{thm:even-family} provide a
length-four code for every \(q\ge2\), while
Theorem~\ref{thm:no-small-length} excludes all lengths \(n\le3\).
\end{proof}

\appendices

\makeatletter
\renewcommand{\theequation}{A\arabic{equation}}
\renewcommand{\thetable}{A\arabic{table}}
\renewcommand{\thefigure}{A\arabic{figure}}
\renewcommand{\thelemma}{A\arabic{lemma}}
\renewcommand{\theremark}{A\arabic{remark}}
\renewcommand{\theproposition}{A\arabic{proposition}}
\renewcommand{\thecorollary}{A\arabic{corollary}}
\renewcommand{\thetheorem}{A\arabic{theorem}}
\setcounter{table}{0}
\setcounter{figure}{0}
\setcounter{lemma}{0}
\setcounter{remark}{0}
\setcounter{corollary}{0}
\setcounter{proposition}{0}
\setcounter{theorem}{0}
\setcounter{equation}{0}

\section{Linear programming setup and nonexistence for \texorpdfstring{$n\le 3$}{n <= 3}}
\label{sec:n-le-3-old}

We now provide an alternate proof that no
permutation-invariant code of dimension \(q\) in $ \Sym^n(\CC^q) $ can detect all single-site errors when \(n\le 3\).  Combined with the explicit
constructions of Sections~\ref{sec:general-construction}
and~\ref{sec:even-construction}, this shows that the family $ ((4,q,2))_q $
has minimal block length for every \(q\ge 2\).

We work in the multiplicity-free intrinsic MacWilliams framework
\cite{usIntrinsicCodes,JustPolynomials,usIntrinsicMacWilliams} for
\[
    V=\Sym^n(\CC^q),
    \qquad
    N:=\dim V=\binom{n+q-1}{n},
\]
in which the conjugation representation on \(\L(V)\) decomposes into sectors
indexed by $ r=0,1,\dots,n.$

A code projector \(P\) of rank \(K\) determines intrinsic projector and twirl
enumerators
\[
    A_r:=A_r(P,P),
    \qquad
    B_r:=B_r(P,P).
\]
Following the coding-theoretic normalization used in
\cite{usIntrinsicCodes,JustPolynomials,usIntrinsicMacWilliams}, define
\[
    \widetilde A_r:=\frac{N}{K^2}A_r,
    \qquad
    \widetilde B_r:=\frac{N}{K}B_r.
\]
We write
\[
    \widetilde{\Avec}
    =
    (\widetilde A_0,\dots,\widetilde A_n)^T,
    \qquad
    \widetilde{\Bvec}
    =
    (\widetilde B_0,\dots,\widetilde B_n)^T.
\]
With this normalization,
\[
    \widetilde A_0=\widetilde B_0=1.
\]

The intrinsic MacWilliams relation becomes
\begin{equation}
\label{eq:macwilliams-nle3}
    \widetilde{\Bvec}=K\,M\,\widetilde{\Avec},
\end{equation}
where \(M\) is the intrinsic MacWilliams matrix for
\(\Sym^n(\CC^q)\).  Equivalently, in the unnormalized intrinsic variables one
has
\[
    \Bvec=M\Avec,
\]
and the scaling above introduces the factor \(K\) in
\eqref{eq:macwilliams-nle3}.

The intrinsic Knill--Laflamme inequalities become
\[
    \widetilde A_r\ge 0,
    \qquad
    \widetilde A_r\le \widetilde B_r
    \qquad
    (0\le r\le n).
\]
If the code detects all single-site errors, then the adjoint sector is
detected, and therefore
\[
    \widetilde A_1=\widetilde B_1.
\]
Thus any permutation-invariant code of dimension \(K\) and distance at least
\(2\) must satisfy
\begin{equation}
\label{eq:d2-constraints-nle3}
    \widetilde A_0=\widetilde B_0=1,
    \quad \widetilde A_1=\widetilde B_1, \quad
     0 \le \widetilde A_r\le \widetilde B_r
    \quad (0\le r\le n)   
\end{equation}
together with \eqref{eq:macwilliams-nle3}.

We now specialize to \(K=q\) and show that this feasibility problem has no
solution for \(n=1,2,3\).  Although the contradiction only uses the trivial
and adjoint rows of the MacWilliams transform, we display the full matrices in
these small cases.

\begin{theorem}
\label{thm:no-small-length-old}
For every \(q\ge 2\), there is no permutation-invariant code of dimension
\(q\) in \(\Sym^n(\CC^q)\) that detects all single-site errors when
\(n\le 3\).  Equivalently, no permutation-invariant code with parameters
\[
    ((n,q,d))_q
\]
and \(d\ge 2\) exists for \(n\le 3\).
\end{theorem}

\begin{proof}
In each of the cases \(n=1,2,3\), infeasibility is detected by the first two
rows of the normalized MacWilliams relation
\[
    \widetilde{\Bvec}=qM\widetilde{\Avec},
\]
together with
\[
    \widetilde A_0=\widetilde B_0=1,
    \qquad
    \widetilde A_1=\widetilde B_1,
    \qquad
    \widetilde A_r\ge 0.
\]

\medskip
\noindent\textbf{Case \(n=1\).}
Here
\[
    N=q,
    \qquad
    M=
    \begin{pmatrix}
    \tfrac{1}{q} & \tfrac{1}{q} \\
    \tfrac{q^2-1}{q} & -\tfrac{1}{q}
    \end{pmatrix}.
\]
Since \(K=q\), the normalized MacWilliams relation
\[
    \widetilde{\Bvec}=qM\widetilde{\Avec}
\]
gives
\[
    \begin{pmatrix}
    \widetilde B_0\\
    \widetilde B_1
    \end{pmatrix}
    =
    \begin{pmatrix}
    1 & 1 \\
    q^2-1 & -1
    \end{pmatrix}
    \begin{pmatrix}
    \widetilde A_0\\
    \widetilde A_1
    \end{pmatrix}.
\]
Thus
\[
    \widetilde B_0=\widetilde A_0+\widetilde A_1,
    \qquad
    \widetilde B_1=(q^2-1)\widetilde A_0-\widetilde A_1.
\]
Using \(\widetilde A_0=\widetilde B_0=1\), the first equation gives
\[
    \widetilde A_1=0.
\]
Substituting into the second equation yields
\[
    \widetilde B_1=q^2-1.
\]
However, adjoint-sector detection requires
\[
    \widetilde B_1=\widetilde A_1.
\]
Thus
\[
    q^2-1=0,
\]
which is impossible for \(q\ge 2\).  Hence no such code exists for \(n=1\).

\medskip
\noindent\textbf{Case \(n=2\).}
Here $N=\frac{q(q+1)}{2}$, and 
\[
    M=
    \smqty(
    \dfrac{2}{q(q+1)} &
    \dfrac{2}{q(q+1)} &
    \dfrac{2}{q(q+1)}
    \\[10pt]
    \dfrac{2(q-1)}{q} &
    \dfrac{q(q+2)-4}{q(q+2)} &
    \dfrac{-4}{q(q+2)}
    \\[10pt]
    \dfrac{q(q-1)(q+3)}{2(q+1)} &
    \dfrac{-q(q+3)}{(q+1)(q+2)} &
    \dfrac{2}{(q+1)(q+2)}
    ).
\]
Since \(K=q\), the normalized MacWilliams relation
\[
    \widetilde{\Bvec}=qM\widetilde{\Avec}
\]
gives
\begin{align*}
    \widetilde B_0
    &=
    \frac{2(\widetilde A_0+\widetilde A_1+\widetilde A_2)}{q+1},
    \\[4pt]
    \widetilde B_1
    &=
    \frac{
    2(q-1)(q+2)\widetilde A_0
    +
    \bigl(q(q+2)-4\bigr)\widetilde A_1
    -
    4\widetilde A_2
    }{q+2}.
\end{align*}
Imposing \(\widetilde A_0=\widetilde B_0=1\), the first equation gives
\begin{equation}
\label{eq:n2-first-M}
    \widetilde A_1+\widetilde A_2=\frac{q-1}{2}.
\end{equation}
Imposing \(\widetilde B_1=\widetilde A_1\), multiplying through by \(q+2\),
and rearranging gives
\begin{equation}
\label{eq:n2-second-M}
    (q-2)(q+3)\widetilde A_1
    -
    4\widetilde A_2
    =
    -2(q-1)(q+2).
\end{equation}
Solving \eqref{eq:n2-first-M}--\eqref{eq:n2-second-M} gives
\[
    \widetilde A_1=-\frac{2(q+1)}{q+2},
    \qquad
    \widetilde A_2=\frac{q^2+5q+2}{2(q+2)}.
\]
Since \(q\ge 2\), we have \(\widetilde A_1<0\), contradicting the required
nonnegativity of the normalized enumerators.  Hence no such code exists for
\(n=2\).

\medskip
\noindent\textbf{Case \(n=3\).}
Here
\[
    N=\frac{q(q+1)(q+2)}{6},
\]
and the intrinsic MacWilliams matrix $M$ is
\[
\resizebox{0.97\columnwidth}{!}{$\displaystyle \smqty(
\tfrac{6}{q(q+1)(q+2)} &
\tfrac{6}{q(q+1)(q+2)} &
\tfrac{6}{q(q+1)(q+2)} &
\tfrac{6}{q(q+1)(q+2)} \\[8pt]
\tfrac{6(q-1)}{q(q+2)} &
\tfrac{2(2q^2+6q-9)}{q(q+2)(q+3)} &
\tfrac{2(q^2+4q-9)}{q(q+2)(q+3)} &
\tfrac{-18}{q(q+2)(q+3)} \\[10pt]
\tfrac{3q(q-1)(q+3)}{2(q+1)(q+2)} &
\tfrac{q(q^2+4q-9)}{2(q+1)(q+2)} &
\tfrac{-2(2q^2+10q-1)}{(q+1)(q+2)(q+4)} &
\tfrac{18}{(q+1)(q+2)(q+4)} \\[10pt]
\tfrac{q(q-1)(q+1)(q+5)}{6(q+2)} &
\tfrac{-q(q+1)(q+5)}{2(q+2)(q+3)} &
\tfrac{2(q+1)(q+5)}{(q+2)(q+3)(q+4)} &
\tfrac{-6}{(q+2)(q+3)(q+4)}
)$}. 
\]
Since \(K=q\), the normalized MacWilliams relation
\[
    \widetilde{\Bvec}=qM\widetilde{\Avec}
\]
gives
\begin{align*}
    \widetilde B_0
    &=
    \tfrac{
    6(\widetilde A_0+\widetilde A_1+\widetilde A_2+\widetilde A_3)
    }{(q+1)(q+2)},
    \\[4pt]
    \widetilde B_1
    &=
    \tfrac{
    6(q-1)(q+3)\widetilde A_0
    +
    2(2q^2+6q-9)\widetilde A_1
    +
    2(q^2+4q-9)\widetilde A_2
    -
    18\widetilde A_3
    }{(q+2)(q+3)}.
\end{align*}
Imposing \(\widetilde A_0=\widetilde B_0=1\), the first equation gives
\begin{equation}
\label{eq:n3-first-M}
    \widetilde A_1+\widetilde A_2+\widetilde A_3
    =
    \frac{(q-1)(q+4)}{6}.
\end{equation}
Using \eqref{eq:n3-first-M} to eliminate \(\widetilde A_3\) from the equation
\(\widetilde B_1=\widetilde A_1\), and clearing denominators, gives
\begin{equation}
\label{eq:n3-second-M}
    (3q-2)(q+3)\widetilde A_1
    +
    2q(q+4)\widetilde A_2
    =
    -3(q-1)(q+2).
\end{equation}
For \(q\ge 2\), the coefficients on the left-hand side are strictly positive:
\[
    (3q-2)(q+3)>0,
    \qquad
    2q(q+4)>0.
\]
The LP constraints require
\[
    \widetilde A_1\ge 0,
    \qquad
    \widetilde A_2\ge 0.
\]
Thus the left-hand side of \eqref{eq:n3-second-M} is nonnegative, whereas the
right-hand side is strictly negative for \(q\ge 2\).  This is a contradiction,
so no such code exists for \(n=3\).

Combining the cases \(n=1,2,3\), we obtain the result.

We note that in each of the cases \(n=1,2,3\), infeasibility is already
detected by the first two normalized MacWilliams equations, the
coding-theoretic normalization
\[
    \widetilde A_0=\widetilde B_0=1,
\]
the adjoint-sector equality
\[
    \widetilde A_1=\widetilde B_1,
\]
and the elementary positivity constraints on the normalized enumerators. 
\end{proof}

As an immediate consequence, the constructions of
Theorems~\ref{thm:general-family} and~\ref{thm:even-family} are minimal in
block length.

\begin{corollary}
\label{cor:minimal-block-length-old}
For every integer \(q\ge 2\), the code constructed in
Theorem~\ref{thm:general-family} when \(q\) is odd and in
Theorem~\ref{thm:even-family} when \(q\) is even has minimal block length
among permutation-invariant codes with parameters
\[
    ((n,q,d))_q
    \qquad
    \text{with } d\ge 2.
\]
Equivalently, four physical qudits are necessary and sufficient for a
distance-two permutation-invariant encoding of one logical qudit of local
dimension \(q\).
\end{corollary}

\begin{proof}
For odd \(q\ge 3\), Theorem~\ref{thm:general-family} gives an explicit
permutation-invariant code with parameters $ ((4,q,2))_q. $
For even \(q\ge 2\), Theorem~\ref{thm:even-family} gives an explicit
permutation-invariant code with the same parameters.  Thus, for every integer
\(q\ge 2\), a length-four code exists.

On the other hand, Theorem~\ref{thm:no-small-length-old} shows that no
permutation-invariant code of dimension \(q\) and distance at least \(2\)
exists for \(n\le 3\).  Therefore \(n=4\) is the smallest possible block
length for every \(q\ge 2\).
\end{proof}


\end{document}